\documentclass[conference]{IEEEtran}
\IEEEoverridecommandlockouts
\usepackage{cite}
\usepackage{amsmath,amssymb,amsfonts}
\usepackage{algorithm}
\usepackage[noend]{algpseudocode}
\usepackage{graphicx}
\usepackage{textcomp}
\usepackage{xcolor}
\usepackage{hyperref}
\usepackage{todonotes}
\usepackage{booktabs}
\usepackage{subcaption}
\usepackage{amsthm}

\newif\iffull
\fulltrue

 \usepackage[font=small,aboveskip=1pt,belowskip=0pt]{caption}
\usepackage{enumitem}

\usepackage{microtype}

\def\BibTeX{{\rm B\kern-.05em{\sc i\kern-.025em b}\kern-.08em
    T\kern-.1667em\lower.7ex\hbox{E}\kern-.125emX}}

\author{\IEEEauthorblockN{Rahul Yesantharao}
\IEEEauthorblockA{\textit{MIT CSAIL}} 
\textit{rahuly@mit.edu}
\and
\IEEEauthorblockN{Yiqiu Wang}
\IEEEauthorblockA{\textit{MIT CSAIL}}
\textit{yiqiuw@mit.edu}
\and
\IEEEauthorblockN{Laxman Dhulipala}
\IEEEauthorblockA{\textit{MIT CSAIL}} 
\textit{laxman@mit.edu}
\and
\IEEEauthorblockN{Julian Shun}
\IEEEauthorblockA{\textit{MIT CSAIL}} 
\textit{jshun@mit.edu}
}

\newcommand{\kdtree}{$k$d-tree}
\newcommand{\kdtrees}{\kdtree s}
\newcommand{\knn}{$k$-NN}
\newcommand{\etal}{{et al}.}
\newcommand{\logtree}{BDL-tree}
\newcommand{\ourtree}{BDL-tree}

\newcommand{\buildveb}{BuildvEB$_\text{S}$}
\newcommand{\buildvebrecursive}{BuildvEBRecursive$_\text{S}$}

\newcommand{\buildbhl}{Build$_\text{S}$}

\newcommand{\eraseS}{Erase$_\text{S}$}
\newcommand{\eraseSrecursive}{EraseRecursive$_\text{S}$}

\newcommand{\insertlog}{Insert$_\text{L}$}
\newcommand{\eraselog}{Erase$_\text{L}$}
\newcommand{\knnlog}{kNN$_\text{L}$}
\newcommand{\knnserial}{kNN$_\text{S}$}

\newcommand{\tablecaption}[1]{\vspace{5pt}\caption{#1}}

\newtheorem{theorem}{Theorem}

\begin{document}

\title{Parallel Batch-Dynamic $k$d-trees}

\maketitle
\thispagestyle{plain} 
\pagestyle{plain} 

\begin{abstract}
$k$d-trees are widely used in parallel databases to support efficient neighborhood/similarity queries. Supporting parallel updates to \kdtrees{} is therefore an important operation.
In this paper, we present BDL-tree, a parallel, batch-dynamic implementation of a \kdtree{} that allows for efficient parallel \knn{} queries over dynamically changing point sets. BDL-trees consist of a log-structured set of $k$d-trees which can be used to efficiently insert or delete batches of points in parallel with polylogarithmic depth. Specifically, given a BDL-tree with $n$ points, each batch of $B$ updates takes $O(B\log^2{(n+B)})$ amortized work and $O(\log(n+B)\log\log{(n+B)})$ depth (parallel time).
We provide an optimized multicore implementation of BDL-trees. Our optimizations include parallel cache-oblivious $k$d-tree construction and parallel bloom filter construction. 

Our experiments on a 36-core machine with two-way hyper-threading  using a variety of synthetic and real-world datasets show that 
our implementation of BDL-tree achieves
a self-relative speedup of up to $34.8\times$ ($28.4\times$ on average) for batch insertions, up to $35.5\times$ ($27.2\times$ on average) for batch deletions, and up to $46.1\times$ ($40.0\times$ on average) for $k$-nearest neighbor queries.
In addition, it achieves throughputs of up to 14.5 million updates/second for batch-parallel updates and 6.7 million queries/second for $k$-NN queries.
We compare to two baseline $k$d-tree implementations and demonstrate that BDL-trees achieve a good tradeoff between the two baseline options for implementing batch updates.

\end{abstract}
\section{Introduction}

Nearest neighbor search is used in a wide range of applications, such as in databases, machine learning, data compression, and cluster analysis. One popular data structure for supporting $k$-nearest neighbor (\knn{}) search in low dimensional spatial data is the \kdtree{}, originally developed by Bentley~\cite{bentley1975}, as it efficiently builds a recursive spatial partition over point sets.

There has been a significant body of work (e.g.,~\cite{bentley1975, bentley-logarithmic-1, bentley-logarithmic-2, pankaj-co, agarwal2016parallel, curtin-faster-dualtree,wehr2018gpu, zellmann2019binned, bkd}) devoted to developing better \kdtree{} variants, both in terms of parallelization and spatial heuristics. However, none of these approaches tackle the problem of parallelizing batched updates, which is important given that many real-world data sets are being frequently updated. In particular, in a scenario  where the set of points is being updated in parallel, existing approaches either become imbalanced or require full rebuilds over the new point set. 
The upper tree in Figure~\ref{fig:baselines} shows the baseline approach that simply rebuilds the \kdtree{} on every insert and delete, maintaining perfect spatial balance but adding overhead to the update operations. On the other hand, the lower tree in Figure~\ref{fig:baselines} shows the other baseline approach, which never rebuilds and instead only inserts points into the existing spatial partition and marks deleted points as tombstones. This gives fast updates at the cost of potentially skewed spatial partitions.

 \begin{figure}[!t]
     \centering
     \includegraphics[width=0.32\textwidth,trim={0 140 250 0},clip]{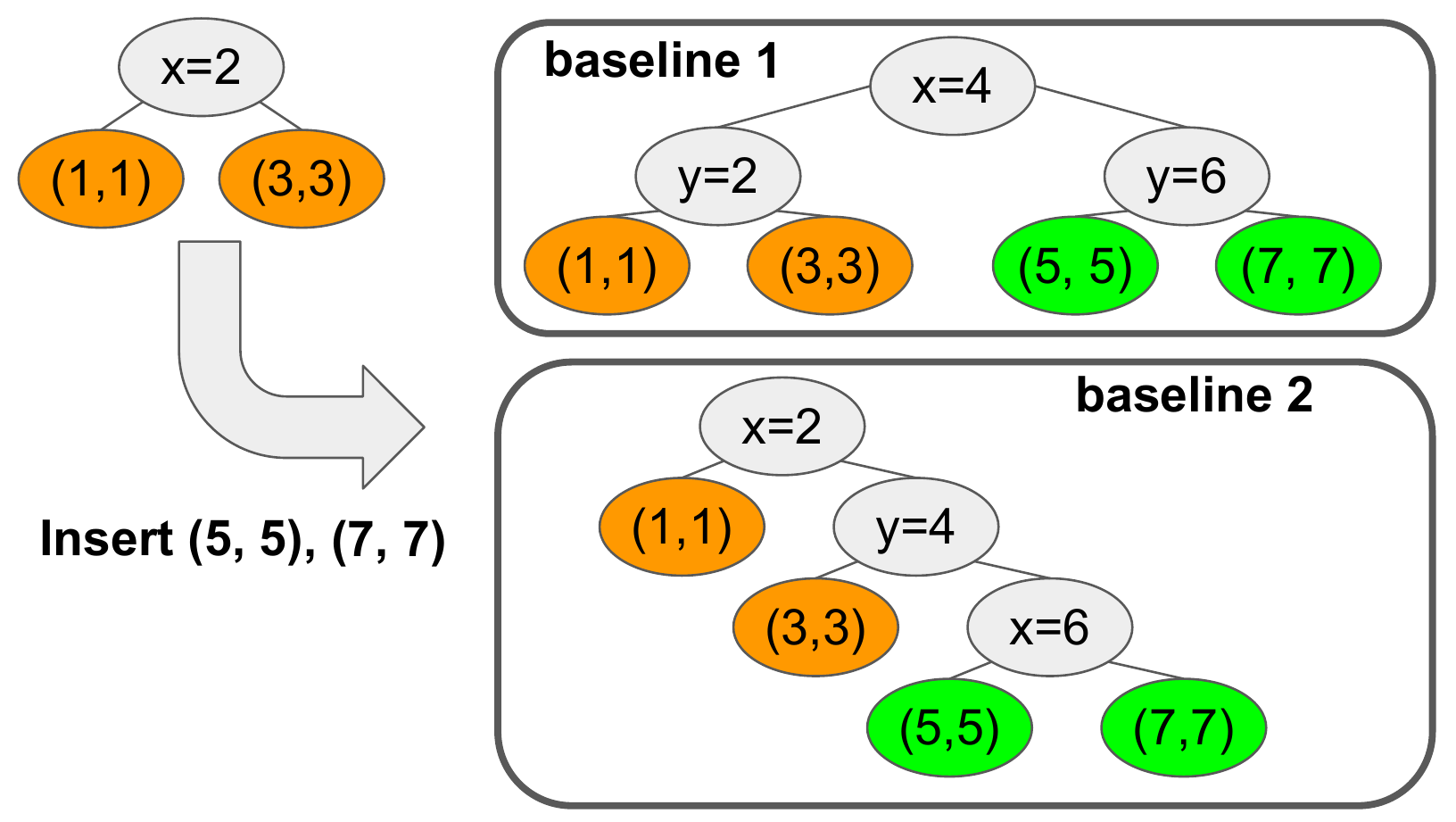}
     \caption{Difference in baseline update strategies when the 2-dimensional points $(5,5)$ and $(7,7)$ are inserted into a spatial-median \kdtree{} initially constructed on the points $(1,1)$ and $(3,3)$. The internal nodes are labeled with the splitting dimension and the coordinate of the split in that dimension. Baseline 1 rebuilds the \kdtree{} on every insertion and deletion, while baseline 2 does not rebuild the tree and instead inserts points into the existing spatial partition.}
     \label{fig:baselines}
 \end{figure}

Procopiuc \etal{}~\cite{bkd} tackled the problem of a dynamically changing point set with their Bkd-Tree, a data structure for maintaining spatial balance in the face of batched updates. 
However, it was developed for the external memory case and is not parallel. Similarly, Agarwal \etal{}~\cite{pankaj-co} developed a dynamic cache-oblivious \kdtree{}, but it was not parallel and did not support batch-dynamic operations.

We adapt these approaches and develop \ourtree{}, a new parallel, in-memory \kdtree{}-based data structure that supports batch-dynamic operations (in particular, batch construction, insertion, and deletion) as well as exact \knn{} queries. 
\ourtree{}s consist of a set of exponentially growing \kdtrees{} and perform batched updates in parallel. This structure can be seen in Figure~\ref{fig:logmethod}. Just as in the Bkd-tree and cache-oblivious tree, our tree structure consists of a small buffer region followed by exponentially growing static \kdtrees{}. Inserts are performed by rebuilding the minimum number of trees necessary to maintain fully constructed static trees. Deletes are performed on the underlying trees, and we rebuild the trees whenever they drop to below half of their original capacity. Our use of parallelism, batched updates, and our approach to maintaining balance after deletion are all distinct from the previous trees that we draw inspiration from. We show that given a \ourtree{} with $n$ points, each batch of $B$ updates takes $O(B\log^2{(n+B)})$ amortized work and $O(\log{(n+B)}\log\log{(n+B)})$ depth (parallel time). As part of our work, we develop, to our knowledge, the first parallel algorithm for the construction of cache-oblivious \kdtrees{}. 
Our construction algorithm takes $O(n\log n)$ work and $O(\log n\log\log n)$ depth. 
In addition, we implement parallel bloom filters to improve the performance of batch updates in practice. We also present a cache-efficient method for performing \knn{} queries in \ourtree{}.
We show theoretically that \ourtree{}s have strong asymptotic bounds on the work and depth of its operations. 

 \begin{figure}[!t]
     \includegraphics[width=0.4\textwidth,trim={0 140 220 0}]{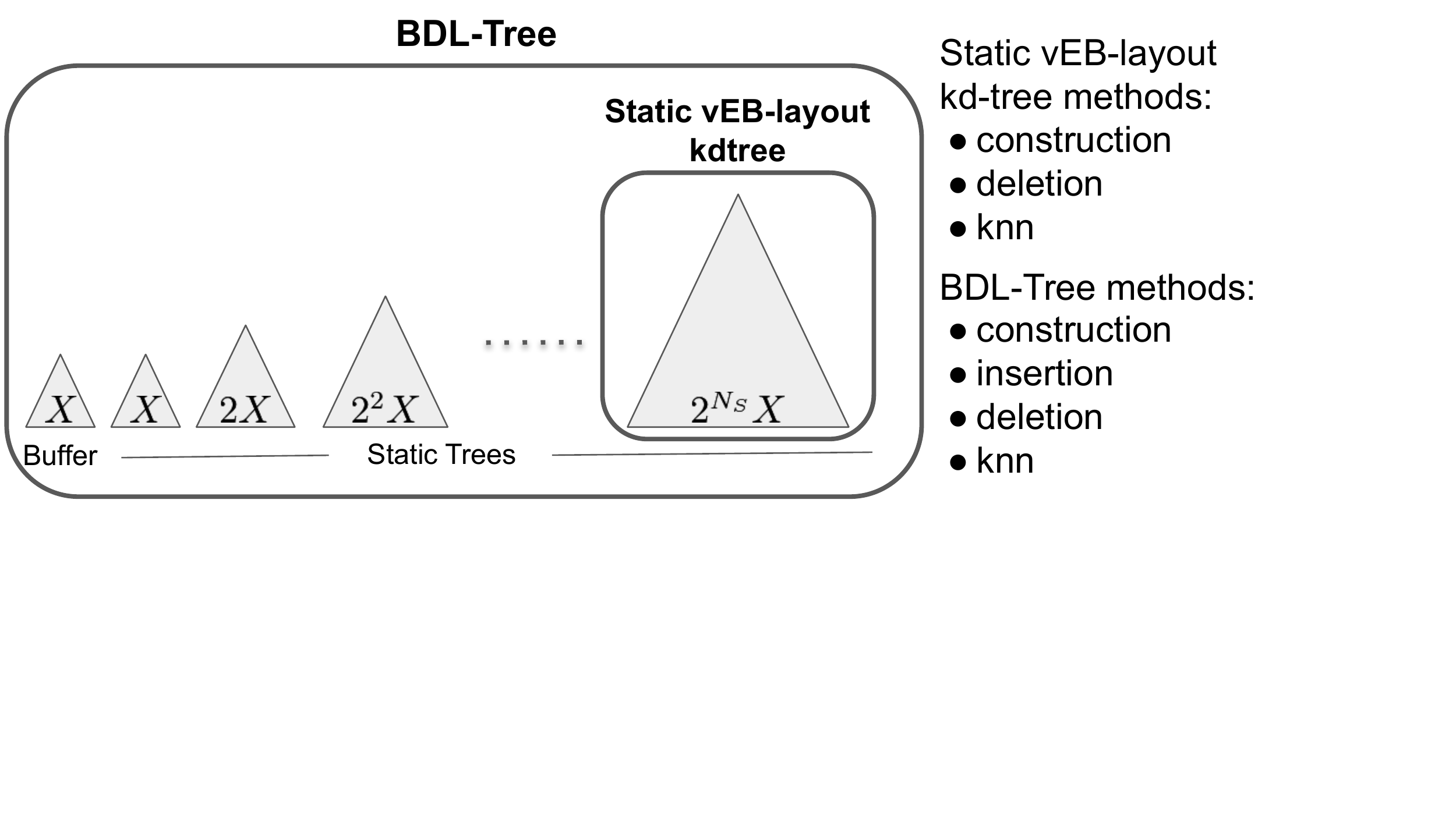}
     \caption{Logarithmic method used in \ourtree{}, with $N_s$ underlying static \kdtrees{} and a buffer \kdtree{} of size $X$. }
     \label{fig:logmethod}
 \end{figure}

We experimentally evaluate \ourtree{}s by designing a set of benchmarks to compare its performance against the two baseline approaches described above, which we implemented using similar optimizations. First, we perform scalability tests for each of the four main operations, construction, batch insertion, batch deletion, and \knn{} in order to evaluate the scalability of our data structure on many cores. On a 36-core machine with two-way hyper-threading, we find that our data structure achieves self-relative speedups of up to $35.4\times$ ($30.0\times$ on average) for construction, up to $35.0\times$ ($28.3\times$ on average) for batch insertion, up to $33.1\times$ ($28.5\times$ on average) for batch deletion, and up to $46.1\times$ ($40.0\times$ on average) for full \knn{}. The largest dataset we test consists of 321 million 3-dimensional points.
Then, we design a set of benchmarks that perform a mixed set of updates and queries in order to better understand the performance of \ourtree{}s in realistic scenarios. 
We find that, when faced with a mixed set of batch operations, \ourtree{}s consistently outperforms the two baselines and presents the best option for such a mixed dynamic setting.
Our source code implementing \ourtree{}s is publicly available at \url{https://github.com/rahulyesantharao/batch-dynamic-kdtree}.

\section{Related Work}
    Numerous tree data structures for spatial search have been proposed
    in the literature, including the \kdtree{}~\cite{bentley1975},
    quadtree~\cite{finkel1974quad}, ball tree~\cite{omohundro1989five},
    cover tree~\cite{beygelzimer2006cover}, and the R-tree~\cite{guttman1984r}.
    Among them, the \kdtree{} is a well-known data structure proposed by
    Bentley~\cite{bentley1975} in 1975.
    It is simple and yet can handle many types of
    queries efficiently. Bentley showed that the \kdtree{} incurs $O(\log n)$
    time for insertion and deletion of random nodes, and logarithmic average
    query time was observed in practice.
    The data structure has been extensively studied in the parallel setting.
    Shevtsov~\cite{shevtsov2007parallelkd}~\etal{} proposed fast construction
    algorithms for the \kdtree{} for ray tracing, by using binning techniques to distribute
    the workload among parallel processors.
    Agarwal~\cite{agarwal2016parallel}~\etal{} proposed parallel algorithms for
    a series of spatial trees, including the \kdtree{} under the massively parallel
    communication (MPC) model.
    Wehr and Radkowski~\cite{wehr2018gpu} proposed fast sorting-based algorithms for
    constructing \kdtrees{} on the GPU.
    Zellmann~\cite{zellmann2019binned} proposed algorithms for CPUs and GPUs
    for \kdtrees{} used in graphics rendering.

    The idea of decomposing a data structure into a logarithmic number of structures
    for the sake of dynamism has been proposed and used in many different scenarios.
    Bentley~\cite{bentley-logarithmic-1, bentley-logarithmic-2} first proposed
    dynamic structures for decomposable search problems.
    Specifically, he proposed general methods for converting static
    data structures into dynamic ones with logarithmic multiplicative overhead in cost,
    using a set of static data structures with size being a power of two.
    Agarwal~\etal~\cite{agarwal2016parallel} designed cache-oblivious
    data structures for orthogonal range searching using ideas from the
    log-structured tree.
    O'Neil~\etal~\cite{oneil1996lsm} proposed the log-structured merge (LSM) tree,
    an efficient file indexing data structure that supports efficient
    dynamic inserts and deletes in systems with multiple storage hierarchies.
    Specifically, batches of updates are cascaded over time, from faster storage
    media to slower ones.
   A parallel version of the LSM tree has been implemented~\cite{parallel-lsm}
    by dividing the key-space independently across cores, and processing it in
    a data-parallel fashion.
    Our work, on the other hand, proposes a parallel batch-dynamic \kdtree{},
    where each batch of updates can be processed efficiently in parallel, while supporting efficient \knn{} queries.
    
    We have recently discovered that, concurrent with our work, Dobson and Blelloch~\cite{dobson2021parallel} developed the zd-tree, a data structure that also supports parallel batch-dynamic updates and \knn{} queries. Their insight is to modify the basic \kdtree{} structure by sorting points based on their Morton ordering and splitting at each level based on a bit in the Morton ordering. 
    For \knn{} queries, they provide a root-based method, which is the traditional top-down \knn{} search algorithm and is what we implement, as well as a leaf-based method, which directly starts from the leaves and searches up for nearest neighbors. The leaf-based method works and is faster when the query points are points in the dataset, as it saves the downward traversal to search for the query points. 
    They prove strong work and depth bounds for construction, batch updates, and \knn{} queries assuming data sets with bounded expansion constant and bounded ratio.  Without these assumptions, however, our bounds would be at least as good as theirs.
    
    In Dobson and Blelloch's implementation, they optimize \knn{} queries by presorting the query points using the Morton ordering to improve cache locality.
    Their algorithm assumes that there is a maximum bounding box where all future data will fit into, while our algorithm does not make this assumption.
    Their implementation discretizes the \texttt{double} coordinates into 64-bit integers in order to perform Morton sort on them. As a result, they can only use at most $64/d$ bits for each of the $d$ dimensions. 
    Directly extending the implementation to higher dimensions would either lead to larger leaf sizes or using larger-width coordinates for 
    the Morton sort, either of which could add overheads. 
    Their experiments consider only 2D and 3D datasets, while we test on higher dimensional datasets as well. 
    For 2D and 3D datasets, the running times that they report seem to be in the same ballpark as our times (their \knn{} queries for constructing the \knn{} graph are much faster than ours due to the presorting described above), after adjusting for number of processors and dataset sizes, but it would be interesting future work to experimentally compare the codes on the same platform and datasets, including higher-dimensional ones. It would also be interesting to integrate some of their optimizations into our code, and to combine the approaches to build a log-structured version of the zd-tree.

\section{Preliminaries}

 \subsection{\kdtree{}}
 The \kdtree{}, first proposed by Bentley~\cite{bentley1975}, is a binary tree data structure that arranges and holds spatial data to speed up spatial queries. Given a set $P$ of $n$ $d$-dimensional points, the \kdtree{} is a balanced binary tree where each node represents a bounding box of a subset of the input points. The root node represents all of the points (and thus the tightest bounding box that includes all the points in $P$). Each non-leaf node holds a splitting dimension and splitting value that splits its bounding box into two halves using an axis-aligned hyperplane in the splitting dimension. Each child node represents the points in one of the two halves.
 This recursive splitting stops when the nodes hold some small constant number of points---these nodes are the leaves and directly represent the points.

 \subsubsection{\knn{} Search using the \kdtree{}}\label{prelim-knn}
 Given a query coordinate $q$, a \knn{} query finds the $k$ nearest neighbors of $q$
 amongst elements of the \kdtree{} by performing a pruned search on the tree~\cite{friedman1977-kdtree-nn}. The canonical approach is to traverse the tree while inserting points in the current node into a buffer that maintains only the
 $k$ nearest neighbors encountered so far. Then, entire subtrees can be pruned during the traversal based on the distance of the $k$'th nearest neighbor found so far.
 
 \subsubsection{Dual-Tree \knn{}}
 Besides the canonical approach explained above, there has been a lot of prior work on ``dual-tree" approaches~\cite{curtin-dualtree, curtin-faster-dualtree, knn-original-dualtree, esmt}, which provide speedups on serial batched \knn{} queries. In particular, the dual-tree approach involves building a second \kdtree{} over the set of query points, and then exploring the two trees simultaneously in order to exploit the spatial partitioning provided by the \kdtree{} structure. We parallelized and tested this dual-tree approach in our work.
 
 \subsection{Batch-Dynamic Data Structures}
 The concept of a parallel batch-dynamic data structure has become popular in recent years~\cite{WangY0S21,dhulipala2021parallel} as an important paradigm due to the availability of large (dynamic) datasets undergoing rapid changes.
 The idea is to batch together operations of a single type and perform them as a single batched update, rather than one at a time. 
 This approach offers two benefits. 
 First, from a usability perspective, it is often the case (especially in applications with a lot of data) that operations on a data structure can be grouped into phases or batches of a single type, so this restriction in the usage of the data structure does not significantly limit the usefulness of the data structure.
 Secondly, from a performance perspective, batching together operations of a single type allows us to group together the involved work and derive significantly more parallelism than otherwise might have been possible (while also avoiding the concurrency issues that might arise with batching together operations of varying types).

\subsection{Baselines}\label{sec:baselines}
 
 In order to benchmark and test the performance of \ourtree{}s, we implement two parallel baseline \kdtrees{} that use opposite strategies for providing batch-dynamism. 
 In particular, the first baseline \kdtree{} simply rebuilds the tree after every batch insertion and batch deletion. 
 With this approach, the tree is able to maintain a fully balanced structure in the face of a dynamically changing dataset, enabling consistently high performance for \knn{} queries. 
 This comes at the cost of reduced performance for updates. 
 The second baseline \kdtree{} never rebuilds the tree and simply maintains the initial spatial partition, inserting points into and deleting points from the existing structure.
 This allows for extremely fast batch insertions and deletions, but could potentially lead to a skewed structure and cause reduced \knn{} performance. The difference between these two baselines is graphically demonstrated in Figure~\ref{fig:baselines}. 
 We demonstrate experimentally that \logtree{} achieves a balanced tradeoff between these two baseline options for batch-dynamic parallel \kdtrees{} on which we are performing \knn{} queries. It outperforms both in the dynamic setting where \knn{} queries and batched updates are all being used.
 
 \subsection{Logarithmic Method}\label{prelim-log-method}

 The logarithmic method~\cite{bentley-logarithmic-1, bentley-logarithmic-2} for converting static data structures into dynamic ones is a very general idea.
 At a high level, the idea is to partition the static data structure into multiple structures with exponentially growing sizes (powers of 2). 
 Then, inserts are performed by only rebuilding the smallest structure necessary to account for the new points. 
 In the specific case of the \kdtree{}, a set of $N_s$ static \kdtrees{} is allocated, with capacities in the set $[2^{0}, 2^{1},\ldots, 2^{N_s-1}]$, as well as an extra buffer tree with size $2^{0}$. 
 Then, when an insert is performed, the insert cascades up from the buffer tree, rebuilding into the first empty tree
 with all the points from the lower trees. 
 If desired, the sizes of all of the trees can be multiplied by a buffer size $X$, which is some constant that is tuned for performance. This structure is illustrated in Figure~\ref{fig:logmethod}. In the figure, all of the trees shown are full; one can imagine that the tree with size $2^{3}X$ is empty, so the next insert would cause the buffer and trees 0, 1, and 2 to cascade up to it.

 \subsection{Computational Model}
 We use the standard work-depth model~\cite{clrs, jaja} for multicore algorithms to analyze theoretical efficiency. The \textbf{work} of an algorithm is the total number of operations used and the \textbf{depth} is the length of the longest sequential dependence (i.e., the parallel running time). An algorithm with work $W$ and depth $D$ can be executed on $p$ processors in $W/p + O(D)$ expected time~\cite{BL99}. Our goal is to come up with parallel algorithms that have low work and depth.
 
 \subsubsection{Parallel Primitives} We use the following parallel primitives in our algorithms.
 \begin{itemize}[topsep=1pt,itemsep=0pt,parsep=0pt,leftmargin=15pt]
     \item \textbf{Prefix sum} takes as input a sequence of values $[a_1,a_2,\ldots,a_n]$, an associative binary operator $\oplus$, and an identity $i$, and returns the sequence $[i,a_1,(a_1\oplus a_2),...,(a_1\oplus a_2\oplus\cdots\oplus a_{n-1})]$ as well as the overall sum of the elements. Prefix sum can be implemented in $O(n)$ work and $O(\log{n})$ depth~\cite{jaja}.
    \item \textbf{Partition} takes an array $A$, a predicate function $f$, and a partition value $p$ and outputs a new array such that all the values $a\in A, a<p$ appear before all the values $a\in A, a\ge p$. Partition can be implemented in $O(n)$ work and $O(\log{n})$ depth~\cite{jaja}.
     \item \textbf{Median partition} takes $n$ elements and a comparator and partitions the elements based on the median value in $O(n)$ work and $O(\log n\log\log n)$ depth~\cite{jaja}. 
 \end{itemize}

\section{\ourtree{}} \label{sec:ourtree}
In this section, we introduce \ourtree{}, a parallel batch-dynamic \kdtree{} implemented using the logarithmic method~\cite{bentley-logarithmic-1, bentley-logarithmic-2} (discussed in Section~\ref{prelim-log-method}). 
\ourtree{}s build on ideas from the Bkd-Tree by Procopiuc \etal{}~\cite{bkd} and the cache-oblivious \kdtree{} by Agarwal \etal{}~\cite{pankaj-co}. The structure is depicted in Figure~\ref{fig:logmethod}.

We implement the underlying \kdtrees{} in an \ourtree{} as nodes in a contiguous memory array, where the root node is the first element in the array.
The \kdtrees{} are built using the van Emde Boas (vEB)~\cite{arge-cache-oblivious-book,demaine2002cache,pankaj-co} recursive layout. 
Agarwal~\cite{pankaj-co}~\etal{} show that this memory layout can be used with \kdtrees{} to make traversal cache-oblivious, although dynamic updates on a single tree become very complex. 
However, in the logarithmic method, the underlying \kdtrees{} themselves are static, and so we are able to sidestep the complexity of cache-oblivious updates on these trees and benefit from the improved cache performance of the vEB layout. 
For the buffer region of the \ourtree{}, we use a regular \kdtree{}, laid out like a binary-heap in memory (i.e., nodes are in a contiguous array, and the children of index $i$ are $2i$, $2i+1$). We will discuss the key parallel algorithms that we used in our implementation: construction, deletion, and \knn{} on the underlying individual \kdtrees{} (Section~\ref{section:single-alg-top}) and construction, insertion, deletion, and \knn{} on \ourtree{} (Section~\ref{section:log-alg-top}). Note that we do not need to support insertions on individual \kdtrees{}, because our \ourtree{} simply rebuilds the necessary \kdtrees{} upon insertions.
We use subscript \textsc{S} to denote algorithms on the underlying \kdtrees{}, and subscript \textsc{L} to denote algorithms on the full logarithmic data structure.

\subsection{Single-Tree Parallel Algorithms}\label{section:single-alg-top}
\subsubsection{Parallel vEB Construction}
\begin{algorithm}[!t]
\caption{Parallel vEB-layout \kdtree{} Construction}\label{alg:veb-construct}
\hspace*{\algorithmicindent} \textbf{Input}: Point Set $P$ \\
\hspace*{\algorithmicindent} \textbf{Output}: \kdtree{} over $P$, laid out with the vEB layout on a contiguous memory array of size $2|P| - 1$.
\begin{algorithmic}[1]
\Procedure{\buildveb{}}{$P$}
    \State Allocate $2|P| - 1$ nodes in contiguous memory. The tree nodes will be laid out in this space.
    \State \textsc{\buildvebrecursive{}($P$, 0, 0, $\lfloor\log(|P|)\rfloor$+1, bottom)}
\EndProcedure
\Procedure{\buildvebrecursive{}}{$Q$, $idx$, $c$, $l$, $t$}
\Statex $idx$: current node index in the memory array
\Statex $c$: current dimension to split on
\Statex $l$: number of levels to build
\Statex $t$: whether we are building the top or bottom of a tree
\State If we hit the base case $n=1$, then we construct a node at $idx$. If $t$ is \textsc{top}, then we perform a parallel median partition on $Q$ in dimension $c$ and record this split as an internal node. Otherwise, we create a leaf node that represents the points in $Q$.
\State Compute $l_b = \left\lceil\left\lceil \frac{l+1}{2} \right\rceil\right\rceil$ and $l_t = l - l_b$ (vEB layout). 
\State Recursively build the top half of the tree with \textsc{\buildvebrecursive{}($Q$, $idx$, $c$, $l_t$, top)}.
\State Compute $idx_b = idx + 2^{l_t}-1$ as the offset where the top half of the tree was just laid out.
\State Construct the $2^{l_t}$ lower subtrees in parallel with \textsc{\buildvebrecursive{}($Q_i$, $idx_i$, $(c+n_t)\mod{d}$, $l_b$, $t$)} where $Q_i$ is the subarray of points that are held by the parent of this subtree and $idx_i$ is the index at which this subtree is to be placed (precomputed with a parallel prefix sum).
\EndProcedure
\end{algorithmic}
\end{algorithm}
The algorithm for parallel construction of the cache-oblivious \kdtree{} is shown in Algorithm~\ref{alg:veb-construct}. The function itself is recursive, and so the top level \textsc{\buildveb{}} function allocates space on line 2 and calls the recursive function \textsc{\buildvebrecursive{}}. Refer to Figure~\ref{fig:vebconstruct} for a graphical representation of this construction.

\begin{figure}
     \centering
     \includegraphics[width=0.48\textwidth,trim={0 160 145 0},clip]{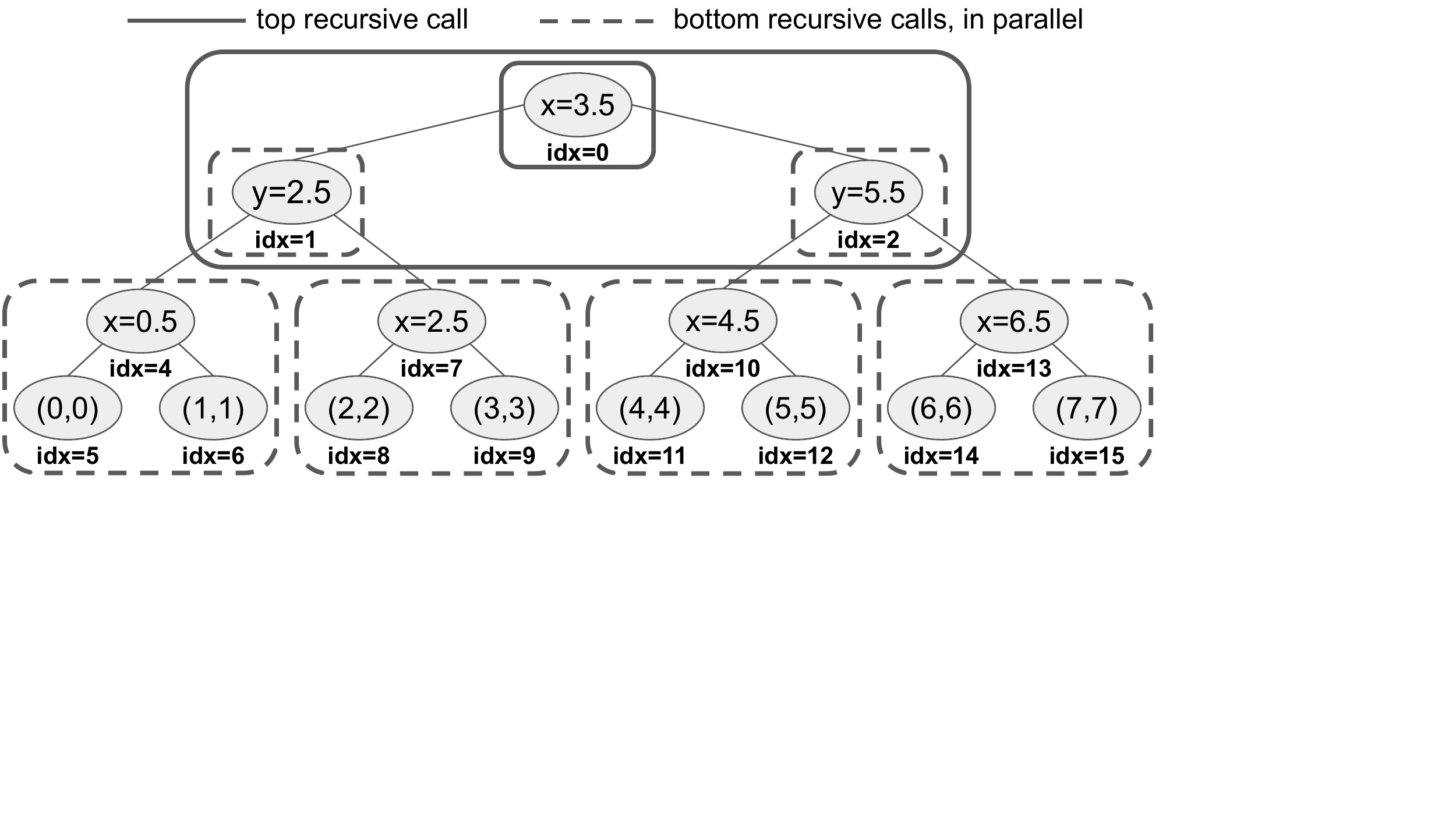}
     \caption{Constructing a vEB \kdtree{} in parallel over 8 2-dimensional points. Note that the top 3 nodes are placed before the remaining 4 bottom subtrees are built in parallel.} 
     \label{fig:vebconstruct}
 \end{figure}

The recursive function \textsc{\buildvebrecursive{}} maintains state with 5 parameters: a point set $Q$, a node index $idx$, a splitting dimension $c$, the number of levels to build $l$, and whether it is building the top or bottom of the tree (indicated by $t$). On line 5, we check for the base case---if the number of levels to build is 1, then we have to construct a node. If this is the top of a tree, then this node will be an internal node, so we perform a parallel median partition in dimension $c$ and save it as an internal node. On the other hand, if this is the bottom of the tree, we construct a leaf node that holds all the points in $Q$.
Lines 6--9 form the recursive step. In accordance with the exponential layout~\cite{pankaj-co}, we have to first construct the top ``half" of the tree and then the bottom ``half". 
Therefore, on line 6, we compute the number of levels $l_b$ in the bottom portion as the hyperceiling\footnote{The hyperceiling of $n$, denoted as $\lceil\lceil n\rceil\rceil$ is the smallest power of 2 that is greater than or equal to $n$, i.e., $2^{\lceil\log{n}\rceil}$.} of $\frac{l+1}{2}$ and the remaining number of levels $l_t$ in the top portion of the tree as $l - l_b$. On line 7, we recursively build the top half of the tree. 
Then, on line 8, we note that because the top half of the tree is a complete binary tree with $l_t$ levels, it will use $2^{l_t} - 1$ nodes. Therefore, we compute $idx_b = idx + 2^{l_t} - 1$, the node index where the bottom half of the tree should start because the trees are laid out consecutively in memory. Finally, on line 9, we construct each of the $2^{l_t}$ subtrees that fall under the top half of the tree, each with $l_b$ levels. Each of these trees falls into a distinct segment of memory in the array, and so we can perform this construction in parallel across all of the subtrees by precomputing the starting index $idx_i$ for each of the $2^{l_t}$ subtrees.

We trace this process on an example in Figure~\ref{fig:vebconstruct}, in which \textsc{\buildveb{}} is called on a set $P$ of 8 points. This spawns a call to \textsc{\buildvebrecursive{}($P[0:8]$, 0, 0, 4, \textsc{bottom})}. On line 6, we will compute $l_b=2$ and $l_t=2$, and on line 7, we spawn a recursive call to \textsc{\buildvebrecursive{}($P[0:8]$, 0, 0, 2, \textsc{top})}. This call is shown as the solid box around the top 3 nodes in Figure~\ref{fig:vebconstruct}. In this call, we will hit one further level of recursion before laying out the 3 nodes in indices 0, 1, 2. Then, the original recursive call will proceed to line 8, where it will compute $idx_b=3$ as the index to begin laying out the $2^{l_t}=4$ bottom subtrees. Finally, on line 9, we will precompute that the starting indices for the 4 bottom subtrees are $(idx_0, idx_1, idx_2, idx_3) = (3, 6, 9, 12)$. This results in 4 parallel recursive calls, shown in the 4 lower dashed boxes in Figure~\ref{fig:vebconstruct}. Each of these recursive calls internally has one more level of recursion to lay out their 3 nodes.

\begin{theorem}\label{thm:veb-construction}
The cache-oblivious \kdtree{} with a vEB layout can be constructed over $n$ points in $O(n\log{n})$ work and $O(\log{n}\log{\log{n}})$ depth.
\end{theorem}
\iffull
\begin{proof}
The work bound is obtained by observing that there are $O(\log{n})$ levels in the fully-constructed tree, and the median partition at each level takes $O(n)$ work, giving a total of $O(n\log{n})$ work. 
For the depth bound, at each recursive step we first build an upper tree with size $O(\sqrt{n})$, and then construct the lower trees in parallel, each with size $O(\sqrt{n})$. 
Further, we use an $O(\log{n})$-depth prefix sum to compute $idx_i$ at every level except the base case and an $O(\log{n}\log\log{n})$-depth median partition in the base case. 
Overall, this results in $O(\log{n}\log{\log{n}})$ depth.
\end{proof}
\fi

After the vEB-layout \kdtree{} is constructed, it can be queried as a regular \kdtree{}---the only difference is the physical layout of the nodes in memory. The correctness of this recursive algorithm can be seen through induction on the number of levels. In particular, we form two inductive hypotheses: 
\begin{itemize}[topsep=1pt,itemsep=0pt,parsep=0pt,leftmargin=15pt]
    \item \textsc{\buildvebrecursive{}($Q$, $idx$, $c$, $l$, \textsc{top})} creates a contiguous, fully-balanced binary tree with $l$ levels rooted at memory location $idx$. Furthermore, this binary tree consists of internal \kdtree{} nodes that equally split the point set $Q$ in half at each level. 
    \item \textsc{\buildvebrecursive{}($Q$, $idx$, $c$, $\lceil\log{|Q|}\rceil+1$, \textsc{bottom})} creates a contiguous \kdtree{} with $l$ levels rooted at memory location $idx$.
\end{itemize}
The base cases, with $l=1$, for these inductive hypotheses are explicitly given on line 5. Then, the inductive step follows easily by noting that the definition of hyperceiling implies that the recursive calls on line 9 are all sized such that $l_b = \lceil\log{|Q|_i}\rceil+1$.

\subsubsection{Parallel Deletion}
\begin{algorithm}[!t]
\caption{Parallel $kd$-Tree Deletion}\label{alg:single-delete}
\hspace*{\algorithmicindent} \textbf{Input}: Point Set $P$
\begin{algorithmic}[1]
\Procedure{\eraseS{}}{$P$}
    \State \textsc{\eraseSrecursive{}($P$, 0)}
\EndProcedure
\Procedure{\eraseSrecursive{}}{$Q$, $idx$}
\Statex $idx$: current node index
\State If the current node is a leaf node, mark any points in the leaf node that are also in $Q$ as deleted. If all of the points in the current leaf are deleted, return \textsc{NULL}. Otherwise, return the current $idx$.
\State Otherwise, perform a parallel partition on $Q$ around the split represented by the current node. Let $Q_{l}, Q_{r}$ be the resulting left and right arrays, respectively, after the partition.
\State Then, recurse on the children in parallel with \textsc{\eraseSrecursive{}($Q_{l}$, $idx_{l}$)}, \textsc{\eraseSrecursive{}($Q_{r}$, $idx_{r}$)}, where $idx_{l}$ and $idx_{r}$ are the IDs of the left and right children, respectively.
\State If neither of the recursive calls return \textsc{NULL}, reset the left and right children to be the results of these calls and return the current node. If both of the recursive calls return \textsc{NULL}, return \textsc{NULL}. 
If one of the recursive calls returns \textsc{NULL} and the other does not, return the non-\textsc{NULL} node.
\EndProcedure
\end{algorithmic}
\end{algorithm}
The algorithm for parallel deletion from a single \kdtree{} is shown in Algorithm~\ref{alg:single-delete}. The function itself is recursive, so the top level \textsc{\eraseS{}} calls the subroutine \textsc{\eraseSrecursive{}} on the root node on line 2.

The recursive function \textsc{\eraseSrecursive{}} acts on one node at a time, represented by the index $idx$. On line 4, it checks for the base case---if the current node is a leaf node, it simply performs a linear scan to mark any points in the leaf node that are also in $Q$ as deleted. 
Then, it returns \textsc{NULL} if the entire leaf was emptied; otherwise, it returns the current node $idx$. Lines 5--7 represent the recursive case. First, on line 5, we perform a parallel partition of $Q$ around the current node's splitting hyperplane. We refer to the lower partition as $Q_{l}$ and the upper partition as $Q_{r}$. On line 6, we recurse on the left and right subtrees in parallel, passing $Q_{l}$ to the left subtree and $Q_{r}$ to the right. Finally, line 7 updates the tree structure. We always ensure that every node has 2 children in order to flatten any unnecessary tree traversal. The return value of \textsc{\eraseSrecursive{}} indicates the node that should take the place of $idx$ in the tree (potentially the same node)---a return value of \textsc{NULL} indicates that the entire subtree rooted at $idx$ was removed. So, if both the left and right child are removed, then we can remove the current node as well by returning \textsc{NULL}. On the other hand, if neither the left or right child are removed, then the subtree is still intact, and we simply reset the left and right child pointers of the current node and return the current node $idx$, indicating that it was not removed. Finally, if exactly one of the children was removed, then we remove the current node as well and let the remaining child connect directly to its grandparent---in this way, we remove an unnecessary internal splitting node. We do this by simply returning the non-\textsc{NULL} child, signaling that it will take the place of the current node in the \kdtree{}. 

\begin{theorem}
Deleting a batch of $B$ points from a single \kdtree{} constructed over $n$ points can be done in $O(B\log{n})$ work and $O(\log{B}\log{n})$ depth in the worst case.
\end{theorem}

\iffull
\begin{proof}
We can see the work bound by noting that each of the $B$ points traverse down $O(\log{n})$ levels as part of the algorithm. For the depth, note that in the worst-case the parallel partition at each level operates over $O(B)$ points at each level. Because parallel partition has logarithmic depth, this would result in a worst-case $O(\log{B})$ depth at each of the $O(\log{n})$ levels, giving the overall depth of $O(\log{B}\log{n})$. 
\end{proof}
\fi

\subsubsection{Data-Parallel \knn{}}

We execute our \knn{} searches in a data-parallel fashion by parallelizing across all of the query points in a batch. The \knn{} search for each point is executed serially.
We implement a ``\knn{} buffer", a data structure that maintains a list of the current $k$-nearest neighbors and provide quick insert functionality to test and insert new points if they are closer than the existing set. The data structure maintains an internal buffer of size $2k$. To insert a point, it simply adds that point to the end of the buffer. If the buffer is filled up, then it uses a serial selection algorithm to partition the buffer around the $k$-th nearest element and clears out the remaining $k$ elements. This achieves a serial amortized $O(1)$ runtime (because the selection partition step is $O(k)$ and is only performed for every $k$ insertions).

To implement batched \knn{} on the \kdtree{}, we perform a \knn{} search for each individual point in parallel across all the points. We now describe the \knn{} method (\textsc{\knnserial{}}) for a single point $p$. We first allocate a \knn{} buffer for the point. Then, we recursively descend through the \kdtree{} searching for the leaf that $p$ falls into. When we find this leaf, we add all of the points in the leaf to the \knn{} buffer. Then, as the recursion unfolds, we check whether the \knn{} buffer has $k$ points. If it does not, we add all the points in the sibling of the current node to the \knn{} buffer to try to fill up the buffer with nearby points as quickly as possible to improve our estimate of the $k$-th nearest neighbor. 
Otherwise, we use the current distance of the $k$-th nearest neighbor to prune subtrees in the tree. In particular, if the bounding box of the current subtree is entirely contained within the distance of the $k$-th nearest neighbor, we add all points in the subtree to the \knn{} buffer. If the bounding box is entirely disjoint, then we prune the subtree. Finally, if they intersect, we recurse on the subtree.

\iffull
\begin{theorem}
 For a constant $k$, \knn{} queries over a batch of $B$ points can be performed over a single \kdtree{} containing $n$ points in worst-case $O(Bn)$ work and worst-case $O(n)$ depth.
\end{theorem}
\begin{proof}
In the worst-case, we have to search the entire tree, of size $O(n)$, resulting in total work of $O(Bn)$ (due to the amortized $O(1)$ insert cost for \knn{} buffers) and depth of $O(n)$, as the queries are done in parallel over the batch, but each search is serial. 
\end{proof}
\fi

As noted by Bentley~\cite{bentley1975} and Friedman \etal{}~\cite{friedman1977-kdtree-nn}, the work for a single nearest-neighbor query on a \kdtree{} is empirically found to be logarithmic in $n$, so the experimental runtime and scalability are much better than suggested by the worst-case bounds.

\subsection{Batch-Dynamic Parallel Algorithms}\label{section:log-alg-top}
This section describes our algorithms for supporting batch-dynamic updates on \ourtree{}s.

\subsubsection{Parallel Insertion}
\begin{algorithm}[!t]
\caption{Parallel \ourtree{} Batch Insertion}\label{alg:log-insert}
\hspace*{\algorithmicindent} \textbf{Input}: Point Set $P$ 
\begin{algorithmic}[1]
\Procedure{\insertlog{}}{$P$}
\State Build an integer bitmask $F$ that represents the static trees within the logarithmic tree structure that are currently filled using 1's, and the trees that are empty using 0's.
\State Compute $F_{new} = F + \frac{|P|}{X}$, where $X$ is the buffer tree size. This is the new bitmask of trees that should be filled.
\State Based on the difference between $F$ and $F_{new}$, determine which trees should be combined into larger trees.
\State Gather the relevant points and construct all the new trees in parallel using \textsc{\buildveb{}} (or \textsc{\buildbhl{}} for the buffer tree).
\EndProcedure
\end{algorithmic}
\end{algorithm}

Insertions are performed in the style of the logarithmic method~\cite{bentley-logarithmic-1, bentley-logarithmic-2}, with the goal of maintaining the minimum number of full trees within \logtree{}. Thus, upon inserting a batch $B$ of points, we rebuild larger trees if it is possible using the existing points and the newly inserted batch. This is implemented as shown in Algorithm~\ref{alg:log-insert}, and depicted in Figure~\ref{fig:bdl-insert}. 

First, on line 2, we build a bitmask $F$ of the current set of full static trees in the logarithmic structure. Then, on line 3, because the buffer \kdtree{} has size $X$, we can add $|P|/X$ to $F$ to compute a new bitmask $F_{new}$ of full trees that would result if we added $|P|$ points to the tree structure. As an implementation detail, note that we first add $|P| \mod{X}$ points to the buffer \kdtree{}---if we fill up the buffer \kdtree{}, then we gather the $X$ points from it and treat them as part of $P$, effectively increasing the size of $P$ by $X$. Then, on line 4, taking the bitwise difference between these two bitmasks gives the set of trees that should be consolidated into new larger trees---specifically, any tree that is set in $F_{new}$ but not in $F$ must be constructed from trees that are set in $F$ but not in $F_{new}$. After determining which trees should be combined into new trees, on line 5 we construct all the new trees in parallel---in parallel for each new tree to be constructed, we deconstruct and gather all the points from trees that are being combined into it and then we construct the new tree over these points and any additional required points from $P$ using Algorithm~\ref{alg:veb-construct}.

Refer to Figure~\ref{fig:bdl-insert} for an example of this insertion method (suppose for this example that $X>2$). In Figure~\ref{fig:bdl-insert-0}, the \ourtree{} contains $X$ points, giving a bitmask of $F=1$ (because only the smallest tree is in use). If we insert $X+1$ points, then we put one node in the buffer tree and compute $F_{new} = 1 + \frac{X}{X} = 2$, and so we have to deconstruct static tree 0 and build static tree 1, as shown in Figure~\ref{fig:bdl-insert-1}. Then, if we insert $X+1$ points again, then we again put one point in the buffer tree and compute $F_{new}=2 + \frac{X}{X} = 3$, and so we simply construct tree 0 on the $X$ new points (leaving tree 1 intact), as seen in Figure~\ref{fig:bdl-insert-2}. Finally, if we then insert $X-1$ points, we note that this would fill the buffer up, so we take 1 point from the buffer and insert $X$ points; then, $F_{new} = 3 + \frac{X}{X} = 4$, and so we deconstruct trees 0, 1 and construct tree 2, as seen in Figure~\ref{fig:bdl-insert-3}.

\begin{figure}[!t]
    \begin{subfigure}{0.24\textwidth}
        \centering
        \includegraphics[width=0.8\textwidth]{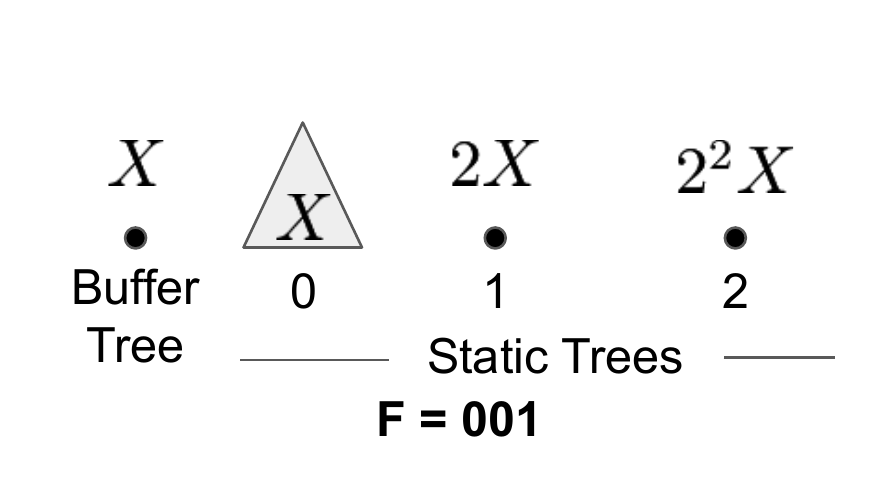}
        \caption{Static tree 0 is full.}
        \label{fig:bdl-insert-0}
    \end{subfigure}
    \begin{subfigure}{0.24\textwidth}
        \centering
        \includegraphics[width=0.8\textwidth]{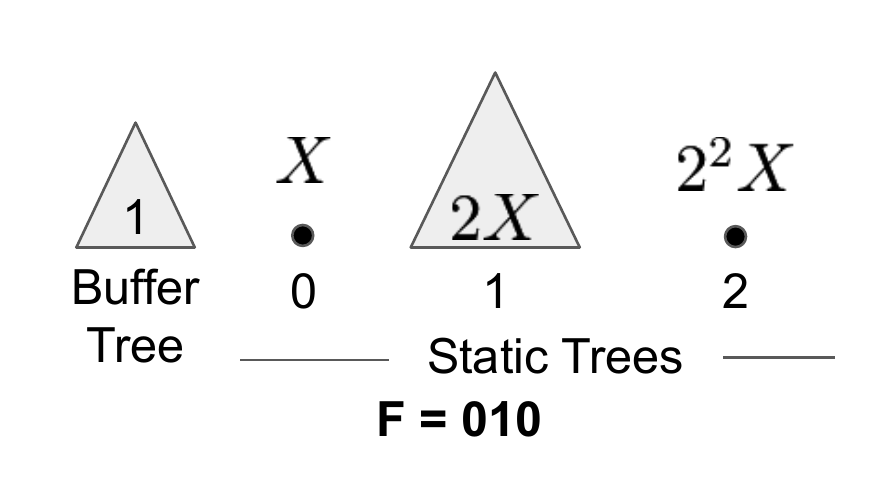}
        \caption{Static tree 1 is full and buffer tree has 1 point.}
        \label{fig:bdl-insert-1}
    \end{subfigure}
    \begin{subfigure}{0.24\textwidth}
        \centering
        \includegraphics[width=0.8\textwidth]{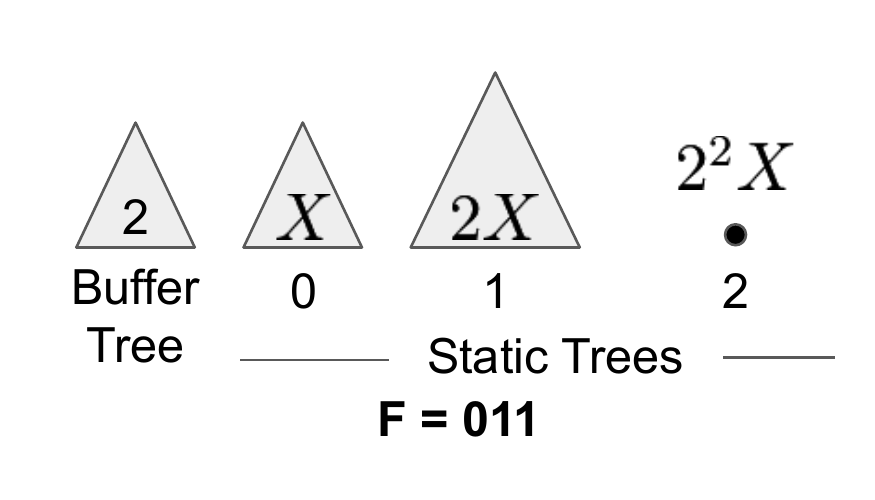}
        \caption{Static trees 0 and 1 are full and buffer tree has 2 points.}
        \label{fig:bdl-insert-2}
    \end{subfigure}
    \begin{subfigure}{0.24\textwidth}
        \centering
        \includegraphics[width=0.8\textwidth]{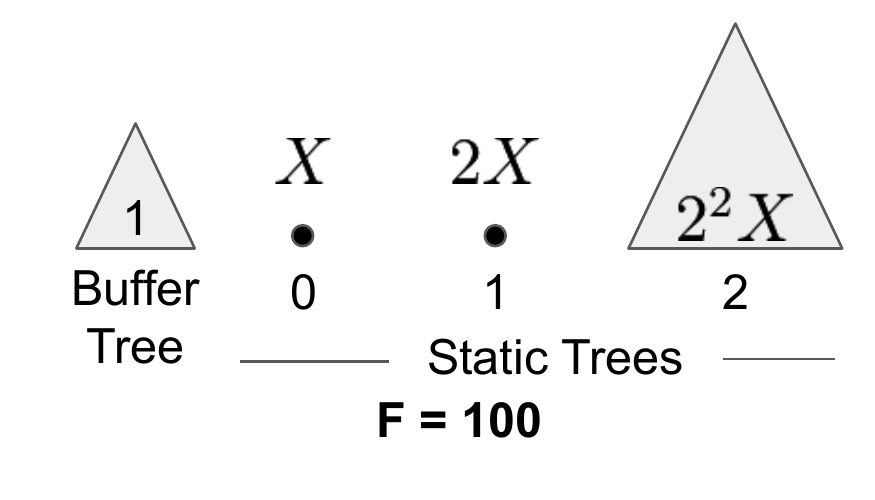}
        \caption{Static tree 2 is full and buffer tree has 1 point.}
        \label{fig:bdl-insert-3}
    \end{subfigure}
    \caption{A \ourtree{} in various configurations with $X>2$; starting from~(a), inserting $X+1$ points gives~(b), then inserting $X+1$ points gives~(c), and then inserting $X-1$ points gives~(d). 
    }
    \label{fig:bdl-insert}
\end{figure}

\subsubsection{Parallel Deletion}
\begin{algorithm}[!t]
\caption{Parallel \ourtree{} Batch Deletion}\label{alg:log-delete}
\hspace*{\algorithmicindent} \textbf{Input}: Point Set $P$
\begin{algorithmic}[1]
\Procedure{\eraselog{}}{$P$}
\State In parallel, delete $P$ from each of the underlying trees which is nonempty by calling \textsc{\eraseS{}($P$)} on each of these trees.
\State In parallel, gather the points from any trees that drop to below half of their original capacity into a set $R$.
\State Call \textsc{\insertlog{}($R$)} to reinsert these points into the log-tree structure.
\EndProcedure
\end{algorithmic}
\end{algorithm}
When deleting a batch of points, the goal is to maintain balance within the subtrees. Thus, if any subtree decreases to less than half of its full capacity, we move all the points down to a smaller subtree in order to maintain balance. As seen in Algorithm~\ref{alg:log-delete}, this is implemented as a three-step process.

On line 2, we call a parallel bulk erase subroutine on each of the individual trees in parallel in order to actually erase the points from the trees. On line 3, we scan the trees in parallel and collect the points from all trees which have been depleted to less than half of their original capacity. Finally, on line 4, we use the \textsc{\insertlog{}} routine to reinsert these points into the structure.

\begin{theorem}
Given an \ourtree{} that was created using only batch insertions and deletions, each batch of $B$ updates takes
$O(B\log^2(n+B))$ amortized work and $O(\log{(n+B)}\log\log{(n+B)})$ depth, where $n$ is the number of points in the tree before applying the updates.

\end{theorem}

\iffull
\begin{proof}
We first argue the work for only performing insertions starting from an empty \ourtree{}. In the worst case, points are added to the structure one by one. Then, as similar to the analysis by Bentley~\cite{bentley-logarithmic-1}, the total work incurred is given by noting that the number of times the $i$'th tree is rebuilt when inserting $m$ points one by one is $O(2^{\log{m}-i})$. Then, summing the total work gives $O(\sum_{i=0}^{\log{m}}{2^{i}i2^{\log{m}-i}}) = O(m\log^2{m})$, where we use the work bound from Theorem~\ref{thm:veb-construction}.  
After inserting a batch of size $B$, we have $n+B$ points in the \ourtree{}, and so the amortized work for the batch is $O(B\log^2 (n+B))$.
Now, if deletions occurred prior to a batch insertion, and the current \ourtree{} has $n$ points, there still must have been $n$ previous insertions (since we started with an empty data structure), and so the work of this batch can still be amortized against those $n$ insertions.  
We now argue the depth bound. When a batch inserted into the tree, the points from smaller trees can be gathered in worst-case $O(\log{(n+B)})$ depth (if all the points must be rebuilt) and the rebuilding process takes worst case $O(\log{(n+B)}\log\log{(n+B)})$ depth, using the result from Theorem~\ref{thm:veb-construction}.

The initial step of deleting the batch of points from each of the underlying \kdtrees{} incurs $O(B\log^2{n})$ work (there are $O(\log{n})$ \kdtrees{}, each taking work $O(B\log{n})$) and depth $O(\log{B}\log{n})$. Then, collecting the points that need to be reinserted can be done in worst-case depth $O(\log{(n+B)})$ and the reinsertion takes $O(\log{(n+B)}\log\log{(n+B)})$ depth, from before.
Overall, the depth is $O(\log{(n+B)}\log\log{(n+B)})$. The amortized work for reinserting points in trees that are less than half full is $O(B\log^2{n})$, as every point we reinsert can be charged to a deletion of another point, either from this batch or from a previous batch. This is because for a tree that is half full, there must be at least as many deletions from the tree as the number of points remaining in the tree.
\end{proof}
\fi

\subsubsection{Data-Parallel $k$-NN}\label{alg:dpknn}
\begin{algorithm}[!t]
\caption{Data-Parallel \ourtree{} \knn{}}\label{alg:log-knn}
\hspace*{\algorithmicindent} \textbf{Input}: Point Set $S$ \\
\hspace*{\algorithmicindent} \textbf{Output}: An array of arrays of the $k$ nearest neighbors of each point in $S$.
\begin{algorithmic}[1]
\Procedure{\knnlog{}}{$S$}
\State Allocate a \knn{} buffer for each of the points in $S$.
\State For each nonempty tree in the \ourtree{}, in serial, call the parallel subroutine \textsc{\knnserial($S$)} on the individual tree, passing the same set of buffers.
\State After all of the individual \textsc{\knnserial{}} calls are complete, gather and return the results from the \knn{} buffers.
\EndProcedure
\end{algorithmic}
\end{algorithm}
In the data-parallel \knn{} implementation, we parallelize over the set of points given to search for nearest neighbors. First, on line 2, we allocate a \knn{} buffer for each of the points in $S$. Then, for each of the non-empty trees in \logtree{}, we call the data-parallel \knn{} subroutine on the individual tree, passing in the set $S$ 
of points and the set of \knn{} buffers. Because we reuse the same set of \knn{} buffers for each underlying \knn{} call, we eventually end up with the $k$-nearest neighbors across all of the individual trees for each point in $S$.

\iffull

\begin{theorem}
For a constant $k$, \knn{} queries over a batch of $B$ points over the $n$ points in the \ourtree{} can be performed in $O(Bn)$ work and $O(n)$ depth.
\end{theorem}
\begin{proof}
These bounds follow directly from the bounds of the underlying individual \knn{} calls. The \knn{} routine on the $i$'th underlying tree, with size $n_i$, has worst-case work $O(Bn_i)$ and depth $O(n_i)$. Summing over all $i$ gives the bounds.  
\end{proof}
\fi

\section{Implementation and Optimizations}
In this section, we describe implementation details and optimizations that we developed to speed up the \ourtree{} in practice.

\subsection{Parallel Bloom Filter}
The bloom filter is a probabilistic data structure for testing set membership, 
which can give false positive matches but no false negative matches~\cite{bloom1970}.
When erasing a batch of points from the \ourtree{}, we need to
potentially search for every point to be deleted within every individual underlying \kdtree{}.
To mitigate this overhead, we use a bloom filter to prefilter points to be erased from each individual \kdtree{} that definitely are not contained within it.
Specifically, we implemented a parallel bloom filter and added one to each individual \kdtree{} within the \ourtree{} to track the points in that individual \kdtree{}. Before erasing an input batch from each individual \kdtree{}, we filter the batch with that \kdtree{}'s bloom filter.
This optimization adds some overhead to the construction and insertion subroutines, but provides significant benefits for the delete subroutine. To mitigate some of this overhead, we construct the bloom filters only when they are needed, rather than during construction.

\subsection{Data-Parallel \knn{} Structure}
We tested three different variants of the data-parallel \knn{} search in order to determine the best parallelization scheme. Each scheme consists of 2 nested for loops; one over the input points and one over the individual \kdtrees{}. In the first variant, the outer loop is over the points and the inner loop is over the \kdtrees{}. We only parallelize the outer loop, so we perform a \knn{} search over all the trees for each point in parallel.
In the second variant, we switch the loop order, so the outer loop is over the \kdtrees{} and the inner loop is over points. We parallelize both the inner and outer loops, so we perform a \knn{} search over each of the trees in a data-parallel fashion (i.e., search for the nearest neighbors of each point in parallel), but we also parallelize over the trees. In this case, because our \knn{} buffer is not thread-safe, for each point we have to allocate a separate buffer for each tree and combine the results at the end, which adds some extra overhead. In the third variant, described in Section~\ref{alg:dpknn}, the outer loop is over the \kdtrees{} and the inner loop is over points, but in contrast to the second scheme, we only parallelize the inner loop over points. So, we perform a data-parallel \knn{} over each of the individual \kdtrees{}, one at a time. We found that the second scheme gives the most parallelism, but the third scheme has the fastest end-to-end running time, both serially and in parallel. This is due to better cache locality of the third method---each tree is completely processed before moving on to the next one. One other optimization we used in this third scheme was to process the underlying trees in order from largest to smallest. We experimentally found that the data-parallel \knn{} search scaled better on large trees, and so this gave better radius estimates earlier (as opposed to processing the trees from smallest to largest).

\subsection{Parallel Splitting Heuristic} \label{sec:splitting-heuristic}
We implemented two different splitting heuristics for our \kdtrees{}: splitting by object median and splitting by spatial median. Object median refers to a true median---at each split, we split around the median of the coordinates of the points in that dimension. On the other hand, for spatial median, we compute the average of the minimum and maximum of the coordinates of the points in the current dimension and use this as the splitting hyperplane. 
The spatial median is faster to compute, but leads to potentially less balanced trees as it is not strictly splitting the points in half at every level. 
We implement the object median with a parallel in-place sort \cite{jaja, parlaylib} and we implement the spatial median in two steps. First, we perform two parallel prefix sums, using $\min()$, $\max()$ as the predicate functions, to compute the minimum and maximum of the values, respectively. Then, we compute the spatial median as the average of these values and perform a parallel partition around this value.

\subsection{Coarsening}
We introduced a number of optimizations that were controlled by tunable parameters in our implementation. One key optimization was the coarsening of leaves, in which we let each leaf node represent up to 16 points rather than 1. This reduces the amount of space needed for node pointers and also improves cache locality when traversing the tree, as the points in a leaf node are stored contiguously.
A second key optimization was the coarsening of the serial base cases of our algorithms, in which we switch from a parallel algorithm to a serial version for recursive cases involving subtrees with less than 1000 points in order to mitigate the overhead of spawning new threads. A third parameter was the buffer tree size, which we set to 1024.

\section{Experiments}

We designed a set of experiments to investigate the performance and scalability of \ourtree{} and compare it to the two baselines described earlier.

\begin{enumerate}[topsep=1pt,itemsep=0pt,parsep=0pt,leftmargin=15pt]
  \item \textbf{B1} is a baseline described in Section~\ref{sec:baselines}, where the \kdtree{} is rebuilt on each batch insertion and deletion in order to maintain balance. This allows for improved query performance (as the tree is always perfectly balanced) at the cost of slowing down dynamic operations.
  \item \textbf{B2} is another baseline described in Section~\ref{sec:baselines}. It inserts points directly into the existing tree structure without recalculating spatial splits. This results in very fast inserts and deletes at the cost of potentially skewed trees (which would slow down query performance).
  \item \textbf{BDL} is our \ourtree{} described in Section~\ref{sec:ourtree}. It represents a tradeoff between batch update performance and query performance. By maintaining a set of balanced trees, it is able to achieve good performance on updates without sacrificing the quality of the spatial partition. 
\end{enumerate}

We use the following set of experiments to measure the scalability of \logtree{} and compare its performance characteristics to the baselines.
\begin{enumerate}[topsep=1pt,itemsep=0pt,parsep=0pt,leftmargin=15pt]
    \item Construction (Section \ref{exp:construction}): construct the tree over a dataset.
    \item Insertion (Section \ref{exp:insert-scalability}): insert fixed-size batches of points into an empty tree until the entire dataset is inserted.
    \item Deletion (Section \ref{exp:delete-scalability}): delete fixed-size batches of points from a tree constructed over the entire dataset until the entire dataset is deleted.
    \item \knn{} (Section \ref{exp:knn-scalability}): find the $k$-nearest neighbors of all the points in the dataset (i.e., compute the \knn{} graph of the dataset).
\end{enumerate}

We also designed the following microbenchmarks in order to better explore the design tradeoffs that \logtree{} makes when compared to the baselines.
\begin{enumerate}[topsep=1pt,itemsep=0pt,parsep=0pt,leftmargin=15pt]
    \item Varying Batch Size (Sections \ref{exp:insert-batch} and \ref{exp:delete-batch}): we vary the batch size of the operations used to fully insert or delete the point set to measure the impact of batch size on throughput.
    \item Varying $k$ after Batched Inserts (Section~\ref{exp:knn-varyk}): we build a tree using a set of batched insertions and then measure the \knn{} search performance with a range of $k$ values to measure the impact of $k$ on throughput and the impact of dynamic updates on \knn{} performance.
    \item Mixed Insert, Delete, and \knn{} Searches (Section \ref{exp:mixed-ops}): we perform a series of batch updates (insertions and deletions) interspersed with \knn{} queries to measure the performance of the data structures over time.
\end{enumerate}

We run all of the experiments over \logtree{} and the two baselines. In addition, for the scalability and batch size experiments, we also compare the object median and spatial median splitting heuristics described in Section~\ref{sec:splitting-heuristic}.

The experiments are all run on an AWS c5.18xlarge instance with 2 Intel Xeon Platinum 8124M CPUs (3.00 GHz), for a total of 36 two-way hyper-threaded cores and 144 GB RAM. Our experiments use all hyper-threads unless specified otherwise. We compile our benchmarks with the \texttt{g++} compiler (version 9.3.0) with the \texttt{-O3} flag, and use ParlayLib~\cite{parlaylib} for parallelism. All reported running times are the medians of 3 runs, after one extra warm-up run for each experiment.

We run the experiments over 9 datasets, consisting of 6 synthetic datasets and 3 real-world datasets. We use two types of synthetic datasets. The first is \textbf{Uniform} (\textbf{U}), consisting of points distributed uniformly at random inside a bounding hyper-cube with side length $\sqrt{n}$, where $n$ is the number of points. The second is \textbf{VisualVar} (\textbf{V}), a clustered dataset with variable-density, produced by Gan and Tao's generator~\cite{gantao2015}. The generator produces points by performing a random walk in a local region, but jumping to random locations with some probability. For each of these two types, we generate them in 2D, 5D, and 7D, and for 10,000,000 points. We also use 3 real-world datasets: \textbf{10D-H-1M}~\cite{ht-dataset, ht-paper} is a 10-dimensional dataset consisting of 928,991 points of home sensor data; \textbf{16D-C-4M}~\cite{chem-dataset, chem-paper} is a 16-dimensional dataset consisting of 4,208,261 points of chemical sensor data; and \textbf{3D-C-321M}~\cite{cosmo} is a 3-dimensional dataset consisting of 321,065,547 points of astronomy data. 
Due to time constraints, we only ran experiments on \textbf{3D-C-321M} using \ourtree{} in parallel to demonstrate that \ourtree{} can scale to large datasets.

\subsection{Construction}\label{exp:construction}
In this benchmark, we measure the time required to construct a tree over each of the datasets. The results using an object median splitting heuristic are shown in Table~\ref{tab:construction-object} and the results using a spatial median splitting heuristic are shown in Table~\ref{tab:construction-spatial}. Figure~\ref{fig:construction} shows the scalability of the throughput on the 10M points 7D uniform dataset.

As we can see from the results, \textbf{BDL} achieves similar or better performance both serially and in parallel than both \textbf{B1} and \textbf{B2}, and has similar or better scalability than both. With the object median splitting heuristic, it achieves up to $34.8\times$ speedup, with an average speedup of $28.4\times$. We also note that the single-threaded runtimes are faster with the spatial-median splitting heuristic than with the object median splitting heuristic. This is expected, because spatial median only involves splitting points at each level compared with finding the median for object-median, hence it is less expensive to compute; however, we also note that the scalability for spatial median is lower because there is less work to distribute among parallel threads.

\begin{table*}

\begin{subtable}{0.5\textwidth}
\centering
\begin{tabular}{@{}l@{\hskip5pt}l@{\hskip5pt}l@{\hskip5pt}l@{\hskip5pt}l@{\hskip5pt}l@{\hskip5pt}l@{}}
\toprule
{} & \multicolumn{3}{@{}l}{1} & \multicolumn{3}{@{}l}{36h} \\
{} &     B1 &     B2 &      BDL &            B1 &           B2 &             BDL \\
\midrule
2D-U-10M &  20.6s &  16.3s &  14.4s &  0.5s (40.0x) &  3.7s (4.5x) &  0.4s (34.5x) \\
2D-V-10M &  20.5s &  16.2s &  14.2s &  0.5s (40.3x) &  3.6s (4.5x) &  0.4s (34.8x) \\
5D-U-10M &  23.3s &  20.8s &  16.3s &  0.7s (35.2x) &  4.8s (4.3x) &  0.5s (30.3x) \\
5D-V-10M &  22.8s &  20.2s &  15.8s &  0.7s (34.9x) &  4.6s (4.4x) &  0.5s (29.2x) \\
7D-U-10M &  27.0s &  24.0s &  17.1s &  0.8s (33.0x) &  5.4s (4.4x) &  0.6s (27.5x) \\
7D-V-10M &  26.2s &  23.2s &  16.5s &  0.8s (33.5x) &  5.3s (4.4x) &  0.6s (26.6x) \\
10D-H-1M &   1.5s &   1.5s &   2.8s &  0.1s (23.7x) &  0.5s (3.4x) &  0.1s (23.8x) \\
16D-C-4M &  13.5s &  14.5s &  11.0s &  0.5s (25.2x) &  3.8s (3.8x) &  0.5s (20.4x) \\
3D-C-321M & -- & -- & -- & -- & -- & 20.4s \\
\bottomrule
\end{tabular}

\vspace{-7pt}
\tablecaption{Object median.}
\label{tab:construction-object}
\end{subtable}

\begin{subtable}{0.5\textwidth}
\centering
\begin{tabular}{@{}l@{\hskip5pt}l@{\hskip5pt}l@{\hskip5pt}l@{\hskip5pt}l@{\hskip5pt}l@{\hskip5pt}l@{}}
\toprule
{} & \multicolumn{3}{@{}l}{1} & \multicolumn{3}{@{}l}{36h} \\
{} &     B1 &    B2 &      BDL &            B1 &           B2 &             BDL \\
\midrule
2D-U-10M &  10.7s &  2.3s &   5.3s &  0.5s (23.8x) &  1.6s (1.4x) &  0.4s (13.5x) \\
2D-V-10M &  10.9s &  2.5s &   5.5s &  0.5s (22.8x) &  1.7s (1.5x) &  0.4s (13.5x) \\
5D-U-10M &  13.7s &  3.4s &  6.7s &  0.8s (18.0x) &  2.3s (1.5x) &  0.6s (10.7x) \\
5D-V-10M &  14.4s &  4.0s &  7.1s &  0.8s (17.4x) &  2.6s (1.5x) &  0.7s (10.3x) \\
7D-U-10M &  16.8s &  4.4s &  7.1s &  1.0s (17.2x) &  2.9s (1.5x) &  0.8s (9.3x) \\
7D-V-10M &  17.3s &  5.2s &  8.3s &  1.0s (16.9x) &  3.3s (1.6x) &  0.9s (9.2x) \\
10D-H-1M &   1.4s &  0.8s &   3.2s &  0.1s (11.0x) &  0.5s (1.6x) &  0.2s (15.9x) \\
16D-C-4M &  10.0s &  5.3s &   7.3s &  0.7s (14.0x) &  3.0s (1.8x) &  0.7s (10.2x) \\
3D-C-321M & -- & -- & -- & -- & -- & 20.8s \\
\bottomrule
\end{tabular}

\vspace{-7pt}
\tablecaption{Spatial median.}
\label{tab:construction-spatial}
\end{subtable}

\caption{Construction times (seconds) for a single thread (1) and 36 cores with hyper-threading (36h). The self-relative speedup for each implementation and dataset is shown in parentheses.}
\end{table*}

\begin{figure*}
    \begin{subfigure}{0.5\textwidth}
        \centering
        \includegraphics[width=0.8\textwidth]{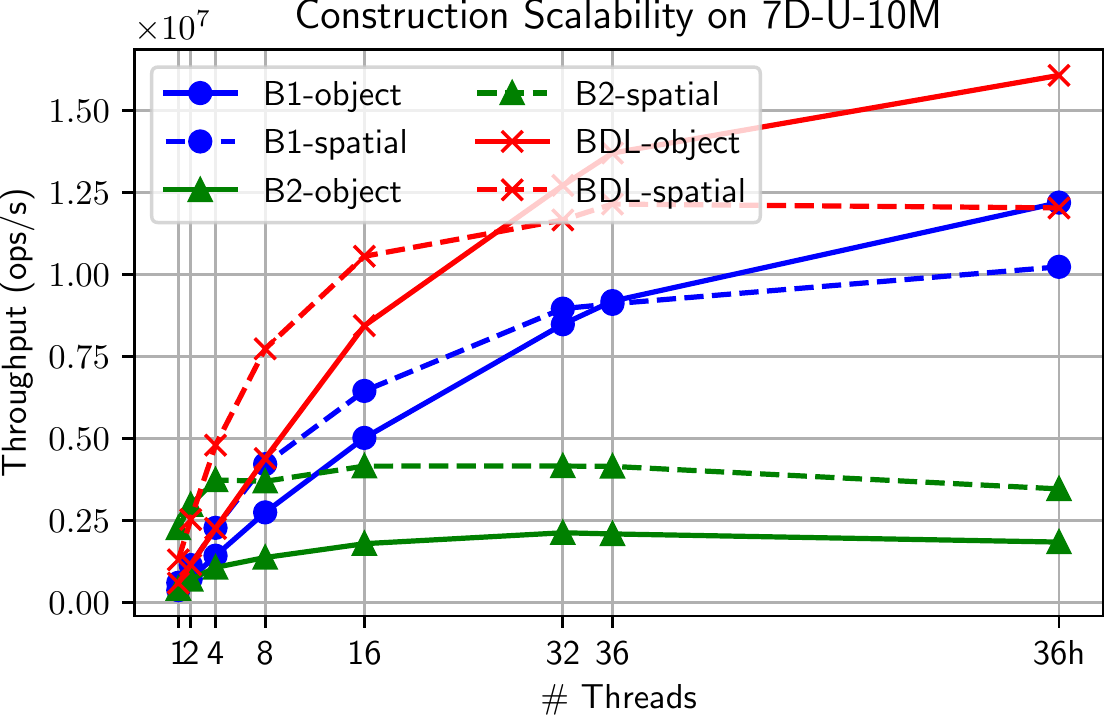}
        \caption{Construction.}
        \label{fig:construction}
    \end{subfigure}
    \begin{subfigure}{0.5\textwidth}
        \centering
        \includegraphics[width=0.8\textwidth]{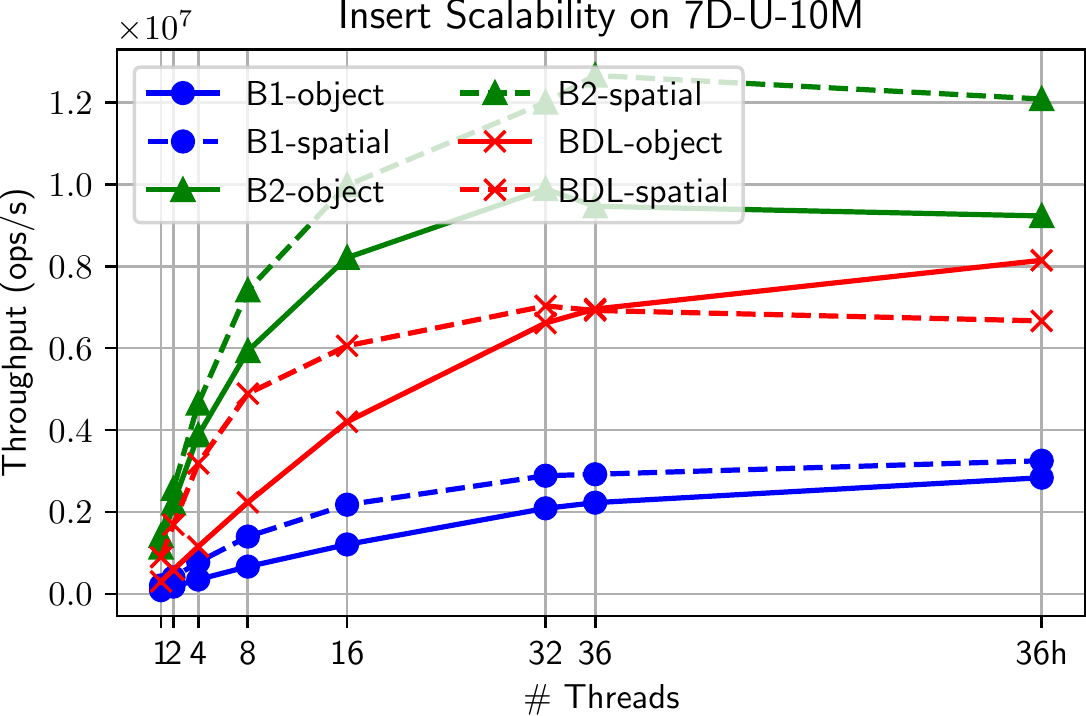}
        \caption{10\% (1M points) Batch Insertion.}
        \label{fig:insert-scalability}
    \end{subfigure}
    \begin{subfigure}{0.5\textwidth}
        \centering
        \includegraphics[width=0.8\textwidth]{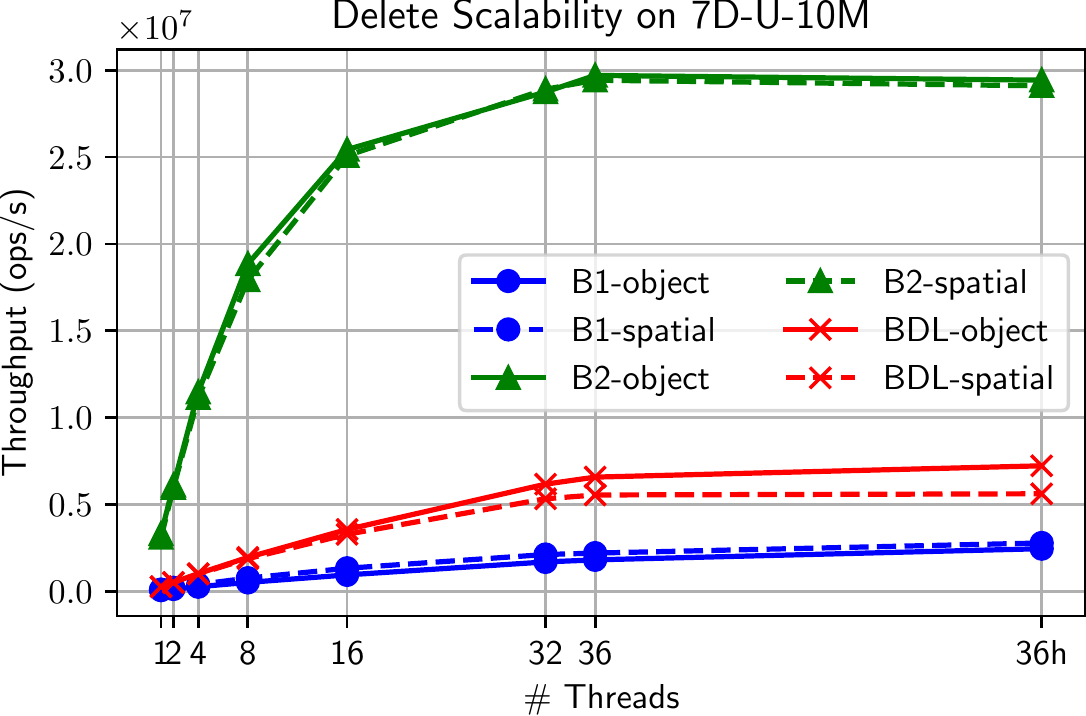}
        \caption{10\% (1M points) Batch Deletion.}
        \label{fig:delete-scalability}
    \end{subfigure}
    \begin{subfigure}{0.5\textwidth}
        \centering
        \includegraphics[width=0.8\textwidth]{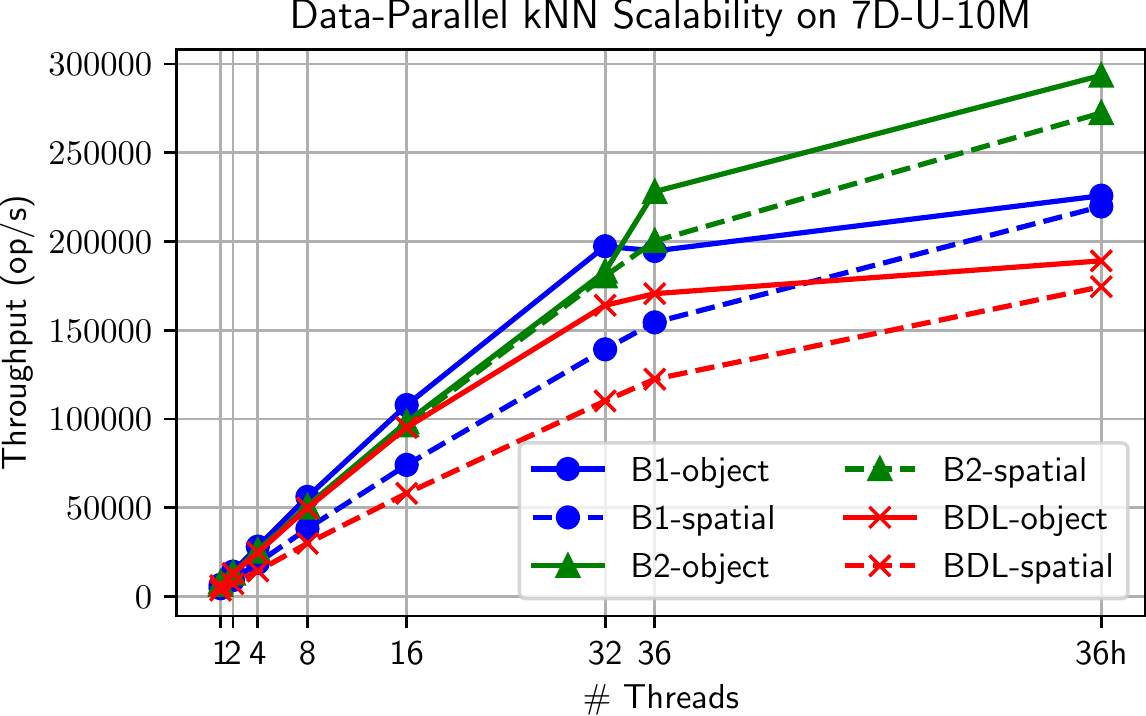}
        \caption{Full (10M points) \knn{} for $k=5$.}
        \label{fig:knn-scalability}
    \end{subfigure}
    \caption{Plot of throughput (operations per second) of batch operations over thread count for both object and spatial median implementations for the 7D-U-10M dataset.}
\end{figure*}

\subsection{Insertion}\label{exp:insertion}
In this benchmark, we measure the performance of our batch insertion implementation as compared to the baselines. We split this experiment into two separate benchmarks, one to measure the full scalability and performance of our implementation, and one to measure the impact of varying batch sizes.

\subsubsection{Scalability}\label{exp:insert-scalability}
In this benchmark, we measure the time required to insert 10 batches each containing 10\% of the points in the dataset into an initially empty tree for each of our two baselines as well as our \ourtree{}. The results using object median and spatial median splitting heuristics are shown in Tables~\ref{tab:insert-object},~\ref{tab:insert-spatial}, respectively. Figure~\ref{fig:insert-scalability} shows the scalability of the throughput on the 10M points 7D uniform dataset.

We see that \textbf{B2} achieves the best performance on batched inserts---this is due to the fact that it does not perform any extra work to maintain balance and simply directly inserts points into the existing spatial structure. \textbf{BDL} achieves the second-best performance---this is due to the fact that it does not have to rebuild the entire tree on every insert, but amortizes the rebuilding work across the batches. Finally, \textbf{B1} has the worst performance, as it must fully rebuild on every insertion. Similar to construction, we note that the spatial median heuristic performs better in the serial case but has lower scalability. With the object median splitting heuristic, \textbf{BDL} achieves parallel speedup of up to $35.5\times$, with an average speedup of $27.2\times$.

\begin{table*}

\begin{subtable}{0.5\textwidth}
\centering
\begin{tabular}{@{}l@{\hskip5pt}l@{\hskip5pt}l@{\hskip5pt}l@{\hskip5pt}l@{\hskip5pt}l@{\hskip5pt}l@{}}
\toprule
{} & \multicolumn{3}{@{}l}{1} & \multicolumn{3}{@{}l}{36h} \\
{} &      B1 &    B2 &     BDL &            B1 &            B2 &             BDL \\
\midrule
2D-U-10M &   86.1s &  5.9s &  24.8s &  2.2s (39.5x) &  0.6s (10.1x) &  0.7s (35.4x) \\
2D-V-10M &   86.1s &  5.9s &  24.4s &  2.2s (39.6x) &  0.6s (10.2x) &  0.7s (35.5x) \\
5D-U-10M &   97.9s &  7.6s &  29.2s &  2.9s (33.6x) &   0.9s (8.7x) &  1.0s (30.3x) \\
5D-V-10M &   94.5s &  7.5s &  28.1s &  2.8s (33.2x) &   0.9s (8.8x) &  1.0s (29.3x) \\
7D-U-10M &  109.7s &  8.8s &  33.0s &  3.5s (31.1x) &   1.1s (8.1x) &  1.2s (26.9x) \\
7D-V-10M &  106.1s &  8.7s &  31.7s &  3.5s (30.7x) &   1.1s (8.2x) &  1.2s (25.6x) \\
10D-H-1M &    7.9s &  0.7s &   1.7s &  0.4s (22.1x) &   0.1s (5.7x) &  0.1s (16.3x) \\
16D-C-4M &   66.2s &  5.5s &  21.1s &  3.0s (22.2x) &   0.9s (6.4x) &  1.1s (18.3x) \\
3D-C-321M & -- & -- & -- & -- & -- & 20.7s \\
\bottomrule
\end{tabular}

\vspace{-7pt}
\tablecaption{Object median.}
\label{tab:insert-object}
\end{subtable}

\begin{subtable}{0.5\textwidth}
\centering
\begin{tabular}{@{}l@{\hskip5pt}l@{\hskip5pt}l@{\hskip5pt}l@{\hskip5pt}l@{\hskip5pt}l@{\hskip5pt}l@{}}
\toprule
{} & \multicolumn{3}{@{}l}{1} & \multicolumn{3}{@{}l}{36h} \\
{} &     B1 &    B2 &      BDL &            B1 &            B2 &             BDL \\
\midrule
2D-U-10M &  32.5s &  4.9s &  5.2s &  1.2s (26.1x) &  0.4s (11.9x) &  0.6s (9.2x) \\
2D-V-10M &  33.9s &  5.0s &  5.8s &  1.4s (24.0x) &  0.5s (10.2x) &  0.6s (9.2x) \\
5D-U-10M &  40.7s &  6.1s &  8.8s &  2.3s (17.5x) &   0.7s (9.4x) &  1.1s (8.1x) \\
5D-V-10M &  44.5s &  6.7s &  10.0s &  2.8s (15.6x) &   0.8s (8.0x) &  1.2s (8.0x) \\
7D-U-10M &  48.1s &  7.0s &  11.2s &  3.1s (15.6x) &   0.8s (8.5x) &  1.5s (7.4x) \\
7D-V-10M &  52.6s &  7.6s &  12.5s &  3.7s (14.4x) &   1.0s (7.8x) &  1.6s (7.6x) \\
10D-H-1M &   6.3s &  0.9s &   1.3s &  0.6s (10.2x) &   0.2s (3.8x) &   0.2s (6.6x) \\
16D-C-4M &  42.9s &  6.0s &  12.9s &  3.6s (11.8x) &   1.0s (5.9x) &  1.5s (8.5x) \\
3D-C-321M & -- & -- & -- & -- & -- & 15.7s \\
\bottomrule
\end{tabular}

\vspace{-7pt}
\tablecaption{Spatial median.}
\label{tab:insert-spatial}
\end{subtable}

\caption{
Batch insertion times (seconds) for a single thread (1) and 36 cores with hyper-threading (36h). The self-relative speedup for each implementation and dataset is shown in parentheses.
We insert batches of 10\% of each dataset, starting from an empty tree until the entire dataset has been inserted.}
\end{table*}

\subsubsection{Batch Size}\label{exp:insert-batch}
In this benchmark, we measure the performance of our batch insertion implementation as the size of the batch varies from 1M points to 5M points. We repeatedly perform batched inserts of the specified size until the entire dataset has been inserted. We provide plots of the results for the 2D VisualVar dataset in Figure~\ref{fig:batch-insert-2dv} and for the 7D Uniform dataset in Figure~\ref{fig:batch-insert-7du}. The first striking result is that the throughput decreases for \textbf{B2} as the batch size increases, while the throughput increases for \textbf{B1} and \textbf{BDL}. This is due to the fact that the work that \textbf{B2} performs increases as the batch size grows---it has to do more work to compute spatial splits over larger batches, whereas for small batches it quickly computes the upper spatial splits on the first batch and never recomputes them. As a result, \textbf{B2} has best throughput at smaller batch sizes, but the worst at large batch sizes. For \textbf{B1} and \textbf{BDL}, note that the throughput increases as the batch size increases. This is due to the fact that each insert has an associated overhead of recomputing spatial partitions, and so the larger the batch size, the fewer times this overhead is paid. For larger batch sizes, \textbf{BDL} has the best throughput among the three implementations for object median. This is again due to the fact that it amortizes the work across inserts, rather than having to recompute spatial splits over the entire dataset at each insert. Finally, note that in most cases, the spatial median heuristic has a better throughput than its object median counterpart, as it takes less work to compute. 

\begin{figure*}
    \begin{subfigure}{0.5\textwidth}
        \centering
        \includegraphics[width=0.8\textwidth]{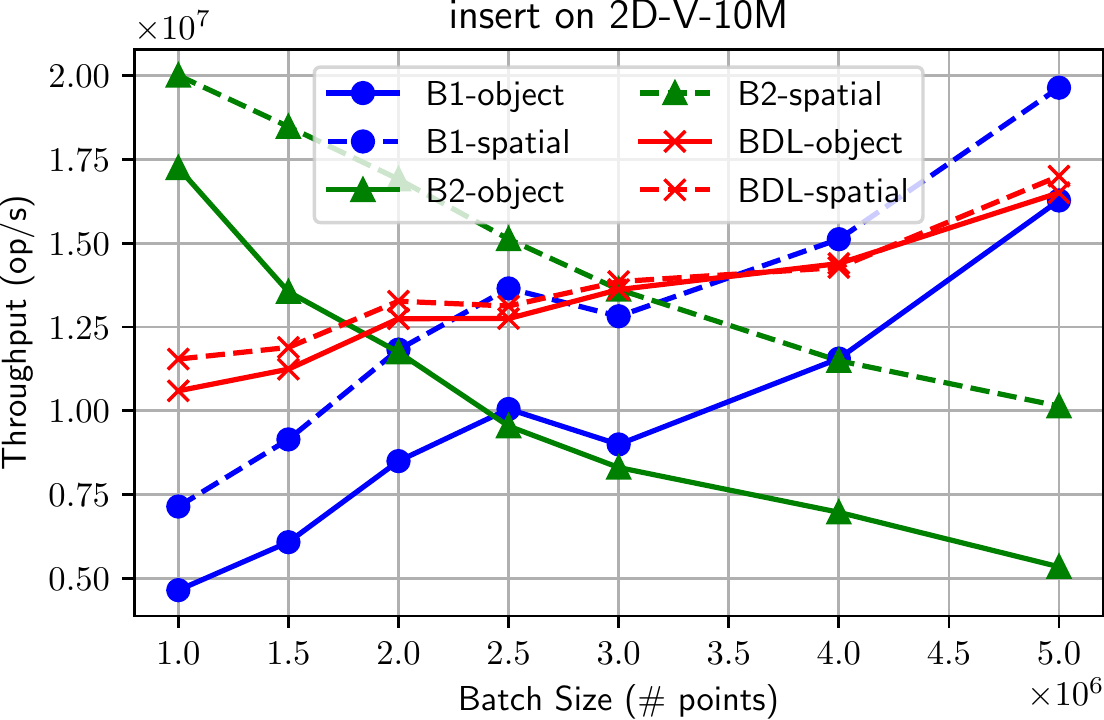}
        \caption{Insertions on 2D-V-10M.}
        \label{fig:batch-insert-2dv}
    \end{subfigure}
    \begin{subfigure}{0.5\textwidth}
        \centering
        \includegraphics[width=0.8\textwidth]{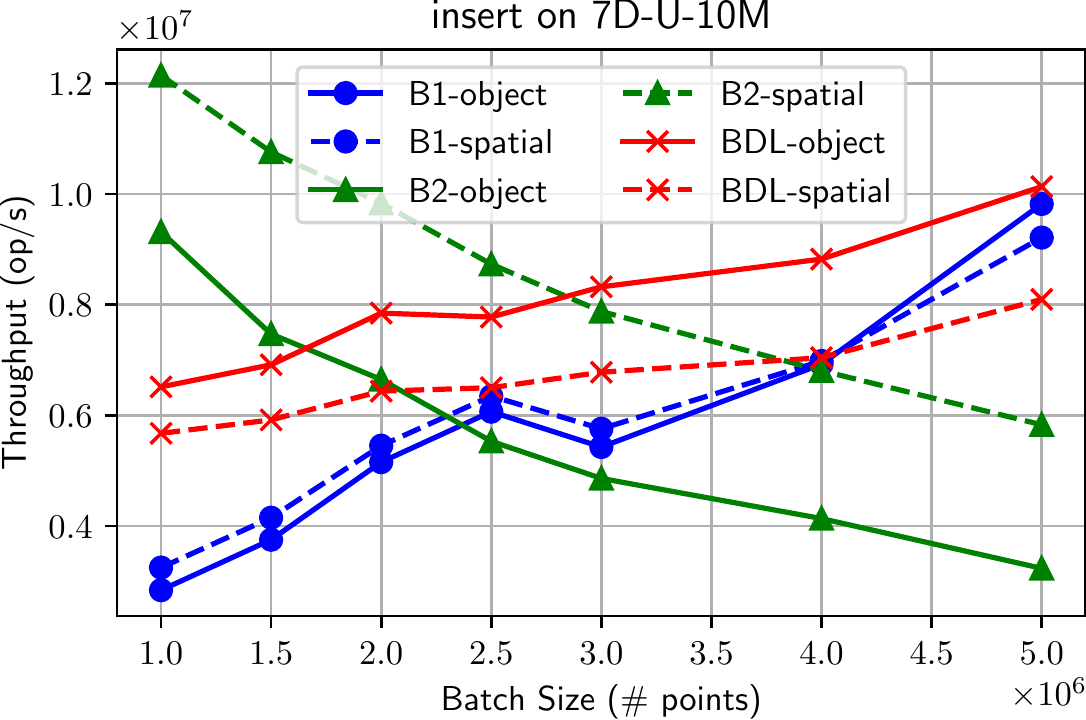}
        \caption{Insertions on 7D-U-10M.}
        \label{fig:batch-insert-7du}
    \end{subfigure}
    \begin{subfigure}{0.5\textwidth}
        \centering
        \includegraphics[width=0.8\textwidth]{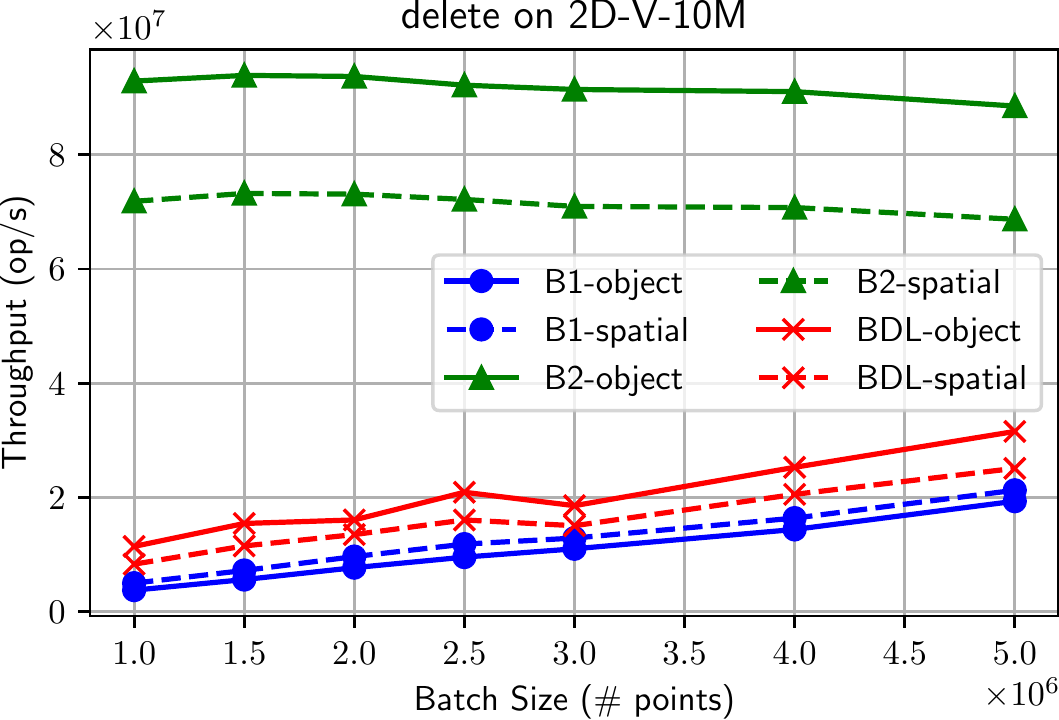}
        \caption{Deletions on 2D-V-10M.}
        \label{fig:batch-delete-2dv}
    \end{subfigure}
    \begin{subfigure}{0.5\textwidth}
        \centering
        \includegraphics[width=0.8\textwidth]{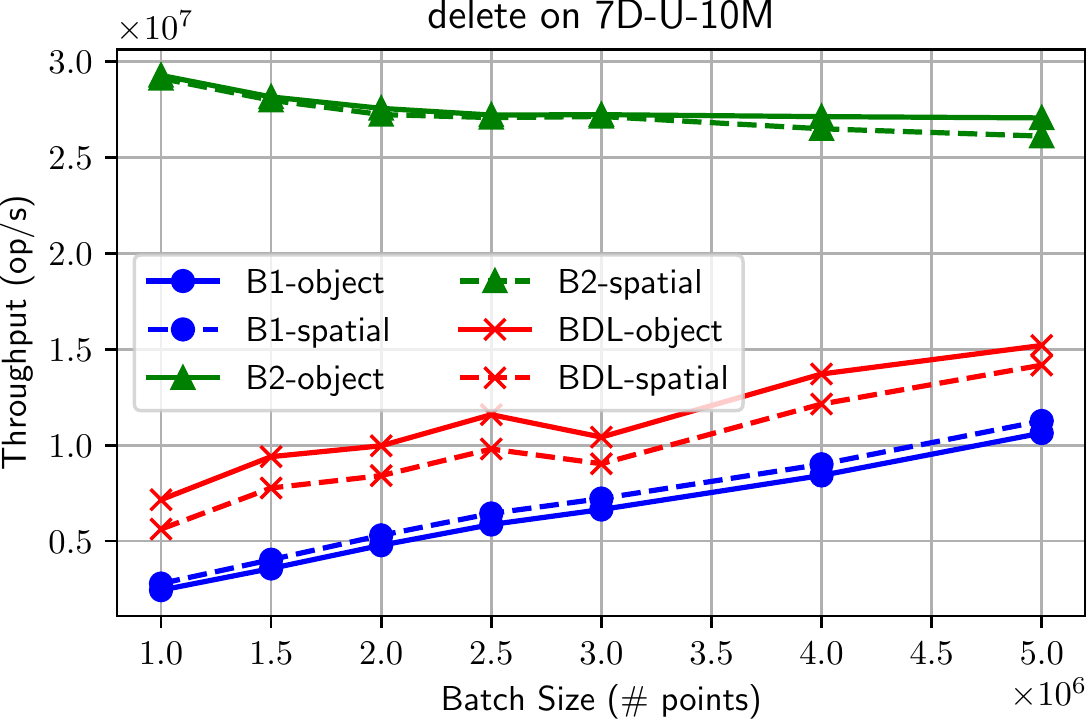}
        \caption{Deletions on 7D-U-10M.}
        \label{fig:batch-delete-7du}
    \end{subfigure}
    \caption{Plot of throughput (operations per second) of batch insertions and batch deletions vs. batch size for the 2D-V-10M and 7D-U-10M datasets.}
\end{figure*}

\subsection{Deletion}\label{exp:deletion}
In this benchmark, we measure the performance of our deletion implementation as compared to the baselines. Similar to the insertion experiment, we split this experiment into two separate benchmarks, one to measure the scalability and performance of our implementation, and one to measure the impact of varying batch size.

\subsubsection{Scalability}\label{exp:delete-scalability}
In this benchmark, we measure the time required to delete 10 batches each containing 10\% of the points in the dataset from an initially full tree for each of our two baselines as well as the \logtree{}. The results using an object median splitting heuristic are shown in Table~\ref{tab:delete-object} and the results using a spatial median splitting heuristic are shown in Table~\ref{tab:delete-spatial}. Figure~\ref{fig:delete-scalability} shows the scalability of the throughput on the 10M point 7D uniform dataset.

We observe that \textbf{B2} has vastly superior performance---it does almost no work other than tombstoning the deleted points so it is extremely efficient. Next, we see that \textbf{BDL} has the second-best performance, as it amortizes the rebuilding across the batches, rather than having to rebuild across the entire point set for every delete. Finally, \textbf{B1} has the worst performance as it rebuilds on every delete. With the object median splitting heuristic, \textbf{BDL} achieves parallel speedup of up to $33.1\times$, with an average speedup of $28.5\times$.

\begin{table*}

\begin{subtable}{0.5\textwidth}
\centering
\begin{tabular}{@{}l@{\hskip5pt}l@{\hskip5pt}l@{\hskip5pt}l@{\hskip5pt}l@{\hskip5pt}l@{\hskip5pt}l@{}}
\toprule
{} & \multicolumn{3}{@{}l}{1} & \multicolumn{3}{@{}l}{36h} \\
{} &      B1 &    B2 &      BDL &            B1 &            B2 &             BDL \\
\midrule
2D-U-10M &  114.0s &  1.1s &  29.2s &  2.7s (43.0x) &  0.1s (10.4x) &  0.9s (33.0x) \\
2D-V-10M &  114.2s &  1.1s &  29.0s &  2.7s (43.0x) &  0.1s (10.5x) &  0.9s (33.1x) \\
5D-U-10M &  130.2s &  2.2s &  40.1s &  3.5s (37.2x) &   0.2s (9.1x) &  1.4s (28.1x) \\
5D-V-10M &  127.0s &  2.2s &  39.3s &  3.5s (36.8x) &   0.2s (9.1x) &  1.4s (27.8x) \\
7D-U-10M &  146.7s &  3.0s &  37.3s &  4.1s (36.0x) &   0.3s (8.9x) &  1.4s (27.0x) \\
7D-V-10M &  144.4s &  3.0s &  36.9s &  4.1s (35.6x) &   0.3s (8.9x) &  1.4s (27.2x) \\
10D-H-1M &   24.7s &  0.2s &  21.9s &  1.0s (24.7x) &   0.0s (5.6x) &  0.8s (28.5x) \\
16D-C-4M &   74.8s &  2.9s &  33.1s &  2.7s (27.4x) &   0.3s (8.6x) &  1.4s (23.6x) \\
3D-C-321M & -- & -- & -- & -- & -- & 17.3s \\
\bottomrule
\end{tabular}

\vspace{-7pt}
\tablecaption{Object median.}
\label{tab:delete-object}
\end{subtable}

\begin{subtable}{0.5\textwidth}
\centering
\begin{tabular}{@{}l@{\hskip5pt}l@{\hskip5pt}l@{\hskip5pt}l@{\hskip5pt}l@{\hskip5pt}l@{\hskip5pt}l@{}}
\toprule
{} & \multicolumn{3}{@{}l}{1} & \multicolumn{3}{@{}l}{36h} \\
{} &      B1 &    B2 &      BDL &            B1 &            B2 &             BDL \\
\midrule
2D-U-10M &   70.2s &  1.3s &  30.2s &  1.8s (38.8x) &  0.1s (11.8x) &  1.1s (26.3x) \\
2D-V-10M &   71.7s &  1.5s &  30.7s &  2.0s (36.0x) &  0.1s (10.7x) &  1.2s (25.4x) \\
5D-U-10M &   82.9s &  2.3s &  29.2s &  2.9s (28.9x) &   0.2s (9.6x) &  1.2s (23.9x) \\
5D-V-10M &   85.9s &  2.9s &  29.6s &  3.4s (25.4x) &   0.3s (9.0x) &  1.3s (22.2x) \\
7D-U-10M &   97.1s &  3.2s &  38.2s &  3.6s (27.1x) &   0.3s (9.4x) &  1.8s (21.4x) \\
7D-V-10M &  100.8s &  3.9s &  38.3s &  4.2s (24.0x) &   0.4s (8.8x) &  1.8s (20.9x) \\
10D-H-1M &   27.6s &  0.4s &  22.1s &  1.1s (25.1x) &   0.1s (4.8x) &  0.8s (26.2x) \\
16D-C-4M &   60.5s &  4.0s &  32.5s &  3.3s (18.3x) &   0.5s (8.8x) &  1.7s (19.5x) \\
3D-C-321M & -- & -- & -- & -- & -- & -- \\
\bottomrule
\end{tabular}

\vspace{-7pt}
\tablecaption{Spatial median.}
\label{tab:delete-spatial}
\end{subtable}

\caption{Batch deletion times (seconds) for a single thread (1) and 36 cores with hyper-threading (36h). The self-relative speedup for each implementation and dataset is shown in parentheses.
We delete batches of 10\% of each dataset, starting from a full tree until the entire dataset has been deleted.}
\end{table*}

\subsubsection{Batch Size}\label{exp:delete-batch}
In this benchmark, we measure the performance of our batch deletion implementation as the size of the batched update varies from 1M points to 5M points. We provide plots of the results for the 2D VisualVar and 7D Uniform datasets in Figures~\ref{fig:batch-delete-2dv},~\ref{fig:batch-delete-7du}, respectively. We note again that \textbf{B2} consistently has the highest throughput, with \textbf{BDL} in second and \textbf{B1} with the lowest throughput. Furthermore, the throughput of \textbf{B1} and \textbf{BDL} increases as the batch size increases; this is true for the same reasons as with insertions. For \textbf{B2}, we observe consistent throughput across batch sizes.

\subsection{Data-Parallel \knn}
In this benchmark, we measure the performance and scalability of our \knn{} implementation as compared to the baselines. We split this into three separate experiments. 

\subsubsection{Scalability}\label{exp:knn-scalability}
In this experiment, we measure the scalability of the \knn{} operation after constructing each data structure over the entire dataset (in a single batch).  The results using an object median splitting heuristic are shown in Table~\ref{tab:knn-object} and the results using a spatial median splitting heuristic are shown in Table~\ref{tab:knn-spatial}. Figure~\ref{fig:knn-scalability} shows the scalability of the throughput on the 10M point 7D uniform dataset. With the object median heuristic, \textbf{BDL} achieves a parallel speedup of up to $46.1\times$, with an average speedup of $40.0\times$.

The results show that \textbf{B1} and \textbf{B2} have similar performance (\textbf{B2} is slightly faster due to implementation differences). Furthermore, they are both faster than \logtree{}. This is to be expected, because the \knn{} operation is performed directly over the tree after it is constructed over the entire dataset in a single batch. Thus, both baselines will consist of fully balanced trees and will be able to perform very efficient \knn{} queries. On the other hand, \textbf{BDL} consists of a set of balanced trees, which adds overhead to the \knn{} operation, as it must be performed separately on each of these individual trees. However, as the next two benchmarks show, \textbf{BDL} provides superior performance in the case of a mixed set of dynamic batch inserts and deletes interspersed with \knn{} queries. 

\begin{table*}

\begin{subtable}{0.5\textwidth}
\centering
\begin{tabular}{@{}l@{\hskip1.5pt}l@{\hskip2.5pt}l@{\hskip2.5pt}l@{\hskip2.5pt}l@{\hskip2.5pt}l@{\hskip1.5pt}l@{}}
\toprule
{} & \multicolumn{3}{@{}l}{1} & \multicolumn{3}{@{}l}{36h} \\
{} &       B1 &       B2 &        BDL &             B1 &             B2 &              BDL \\
\midrule
2D-U-10M &    34.9s &    11.9s &    64.5s &   0.6s (57.2x) &   0.3s (40.3x) &   1.5s (43.1x) \\
2D-V-10M &    37.2s &    12.7s &    67.0s &   0.7s (57.0x) &   0.3s (40.7x) &   1.5s (43.7x) \\
5D-U-10M &   302.3s &   178.0s &   339.4s &   6.3s (47.6x) &   3.5s (51.4x) &   8.6s (39.4x) \\
5D-V-10M &   109.5s &    64.1s &   145.8s &   2.1s (51.1x) &   1.2s (52.7x) &   3.2s (46.1x) \\
7D-U-10M &  1520.0s &  1239.4s &  1621.6s &  44.3s (34.3x) &  34.1s (36.4x) &  52.9s (30.7x) \\
7D-V-10M &   133.6s &    86.3s &   173.1s &   2.8s (48.1x) &   1.7s (50.7x) &   4.0s (43.8x) \\
10D-H-1M &     5.3s &     5.5s &    11.9s &   0.1s (54.5x) &   0.1s (50.5x) &   0.3s (45.8x) \\
16D-C-4M &   464.4s &   612.2s &   468.3s &  16.5s (28.1x) &  16.1s (38.1x) &  17.3s (27.1x) \\
3D-C-321M & -- & -- & -- & -- & -- & 15.6s \\
\bottomrule
\end{tabular}

\vspace{-7pt}
\tablecaption{Object median.}
\label{tab:knn-object}
\end{subtable}

\begin{subtable}{0.5\textwidth}
\centering
\begin{tabular}{@{}l@{\hskip1.5pt}l@{\hskip2.5pt}l@{\hskip2.5pt}l@{\hskip2.5pt}l@{\hskip2.5pt}l@{\hskip1.5pt}l@{}}
\toprule
{} & \multicolumn{3}{@{}l}{1} & \multicolumn{3}{@{}l}{36h} \\
{} &       B1 &       B2 &        BDL &             B1 &             B2 &              BDL \\
\midrule
2D-U-10M &    33.6s &    13.0s &    69.9s &   0.6s (55.6x) &   0.3s (40.0x) &   1.6s (44.9x) \\
2D-V-10M &    35.7s &    13.5s &    72.7s &   0.6s (56.8x) &   0.3s (39.5x) &   1.6s (44.9x) \\
5D-U-10M &   348.8s &   216.7s &   503.1s &   6.7s (52.2x) &   4.0s (54.1x) &   9.6s (52.2x) \\
5D-V-10M &   103.1s &    65.2s &   174.2s &   2.0s (52.8x) &   1.4s (45.4x) &   3.5s (49.5x) \\
7D-U-10M &  2142.0s &  1494.0s &  2804.0s &  45.5s (47.1x) &  36.7s (40.7x) &  57.3s (48.9x) \\
7D-V-10M &   116.9s &    81.3s &   203.1s &   2.3s (50.5x) &   1.7s (48.0x) &   4.1s (49.3x) \\
10D-H-1M &     5.7s &     4.7s &    15.5s &   0.1s (54.6x) &   0.1s (45.2x) &   0.3s (47.3x) \\
16D-C-4M &   606.7s &   557.0s &   623.4s &  19.9s (30.5x) &  11.9s (47.0x) &  20.4s (30.5x) \\
3D-C-321M & -- & -- & -- & -- & -- & 17.4s \\
\bottomrule
\end{tabular}

\vspace{-7pt}
\tablecaption{Spatial median.}
\label{tab:knn-spatial}
\end{subtable}

\caption{
\knn{} times (seconds) for a single thread (1) and 36 cores with hyper-threading (36h). The self-relative speedup for each implementation and dataset is shown in parentheses.
We use 100\% of dataset except for 3D-C-321M, which was run with 10\% of the dataset because the full \knn{} results could not fit in memory. }
\end{table*}

\subsubsection{Effect of Varying $k$}\label{exp:knn-varyk}
In this experiment, we benchmark the throughput of the \knn{} operation on 36 cores with hyper-threading as $k$ varies from 2 to 11. For all three trees, we perform the \knn{} operation after building the tree from a set of batch insertions, with a batch size of 5\% of the dataset, until the entire dataset is inserted. The results are shown for the 2D VisualVar dataset in Figure~\ref{fig:varyk-knn-2dv} and for the 7D Uniform dataset in Figure~\ref{fig:varyk-knn-7du}. In both scenarios, we see that \textbf{B1} has the best \knn{} performance, followed closely by \textbf{BDL}. \textbf{B2} has significantly worse performance---this is because the construction of the tree was performed with a set of batch insertions, rather than a single construction over the entire dataset, the tree ends up imbalanced and the \knn{} query performance suffers.

\begin{figure*}
    \begin{subfigure}{0.5\textwidth}
        \centering
        \includegraphics[width=0.8\textwidth]{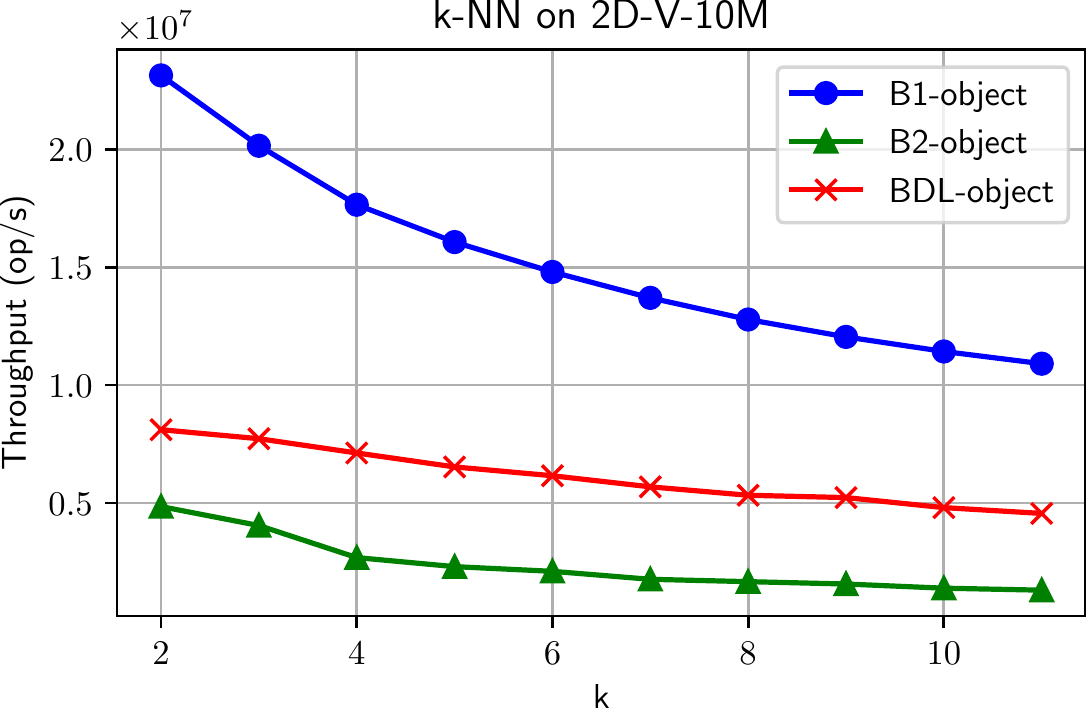}
        \caption{\knn{} on 2D-V-10M.}
        \label{fig:varyk-knn-2dv}
    \end{subfigure}
    \begin{subfigure}{0.5\textwidth}
        \centering
        \includegraphics[width=0.8\textwidth]{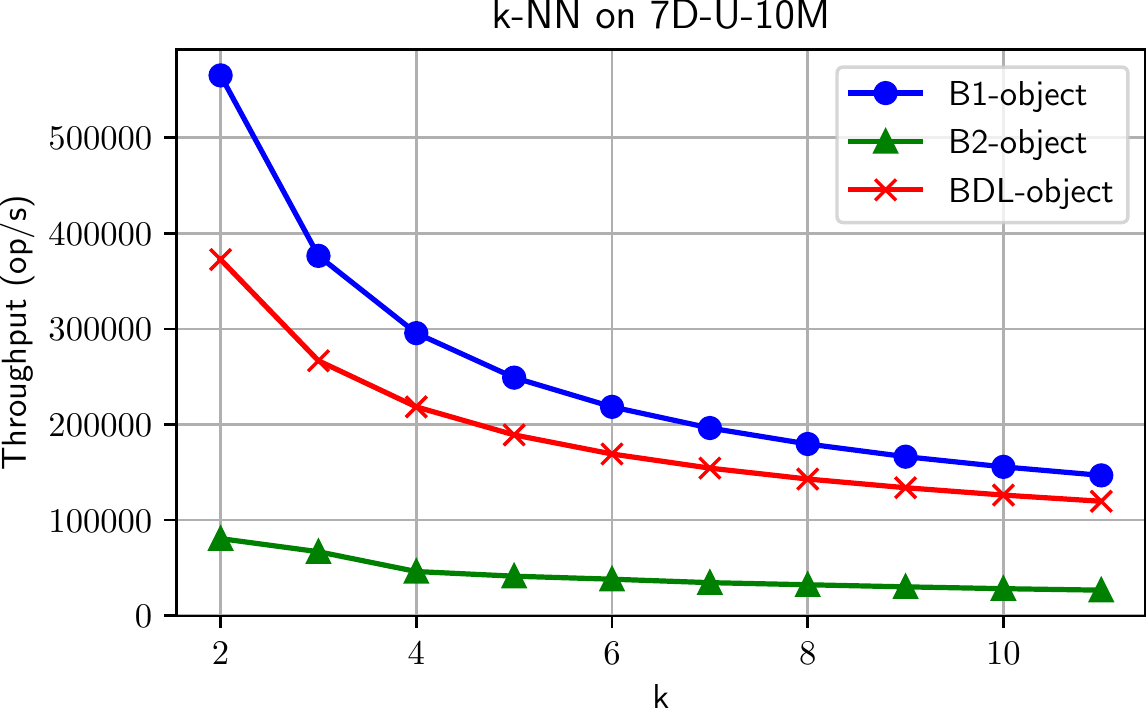}
        \caption{\knn{} on 7D-U-10M.}
        \label{fig:varyk-knn-7du}
    \end{subfigure}
    \caption{Plots of \knn{} throughput (operations per second) vs. $k$ using all 36h cores with hyperthreading,
    for the 2D-V-10M and 7D-U-10M datasets.}
\end{figure*}

\subsubsection{Mixed Operations}\label{exp:mixed-ops}
In this final experiment, we measure the \knn{} and overall performance of the trees as a mixed set of batch insertions, batch deletions, and batch \knn{} queries are performed. In particular, we perform a set of 20 batch insertions, each consisting of 5\% of the dataset, into the tree. After every 5 batch insertions, we perform a \knn{} query with $k=5$ (over the entire dataset) to measure the current query performance. Then, we perform a set of 15 batch deletes, each consisting of a random 5\% of the dataset (with no repeats). After every 5 batch deletes, we again perform a \knn{} ($k=5$) query. Overall, there are 7 \knn{} queries. We thus split the experiment into 7 distinct sections, demarcated by the \knn{} queries. For each section (labeled as INS0, INS1, INS2, INS3, DEL0, DEL1, DEL2), we measure the \knn{} runtime as well as the time for the 5 batched insertion/deletion operations. The results are shown for the 2D VisualVar dataset in Figure~\ref{fig:dynamic-ops-2dv} and for the 5d Uniform dataset in Figure~\ref{fig:dynamic-ops-5du} (we observed similar results for other datasets). The $x$-axis shows the 7 sections, and the $y$-axis shows the time for each of these sections. There are two lines for each tree---a dashed line indicating just the \knn{} times at each section, and a solid line indicating the total time for the batched update and \knn{}. 

In the Uniform case, we see that \textbf{B2} performs the worst overall, due almost entirely to its poor \knn{} performance after batched updates. Note that the batched updates themselves contribute minimally to the total runtime of this baseline---they are very fast but cause significant imbalance in the tree structure, leading to degraded query performance. \textbf{B1} has the best raw \knn{} query time, but its overall runtime is the second worst, as the batched updates are quite expensive (in order to maintain the balance that results in fast query times). Finally, \textbf{BDL} represents the best tradeoff between dynamic batch updates and \knn{} performance. In particular, it has the best total runtime after every operation. The results are quite similar for the VisualVar case, except that we note that \textbf{B1} has the overall worst runtime, due to very expensive updates.

\begin{figure*}
    \begin{subfigure}{0.5\textwidth}
        \centering
        \includegraphics[width=0.8\textwidth]{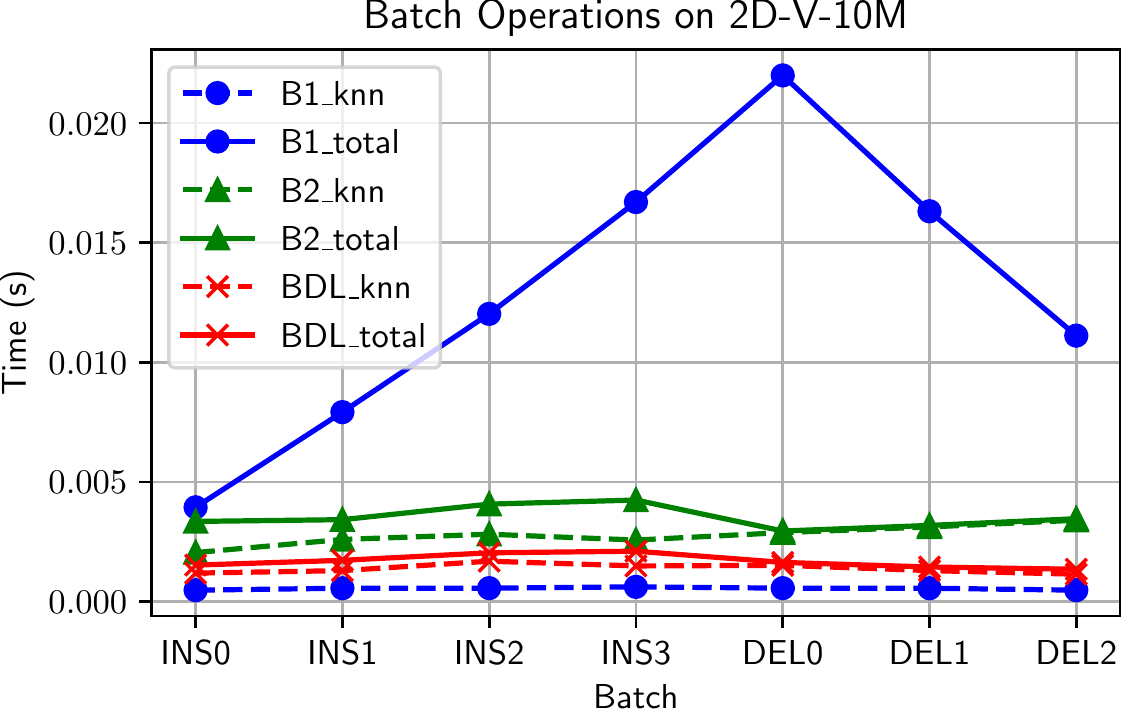}
        \caption{Batch updates and \knn{} queries on 2D-V-10M.}
        \label{fig:dynamic-ops-2dv}
    \end{subfigure}
    \begin{subfigure}{0.5\textwidth}
        \centering
        \includegraphics[width=0.8\textwidth]{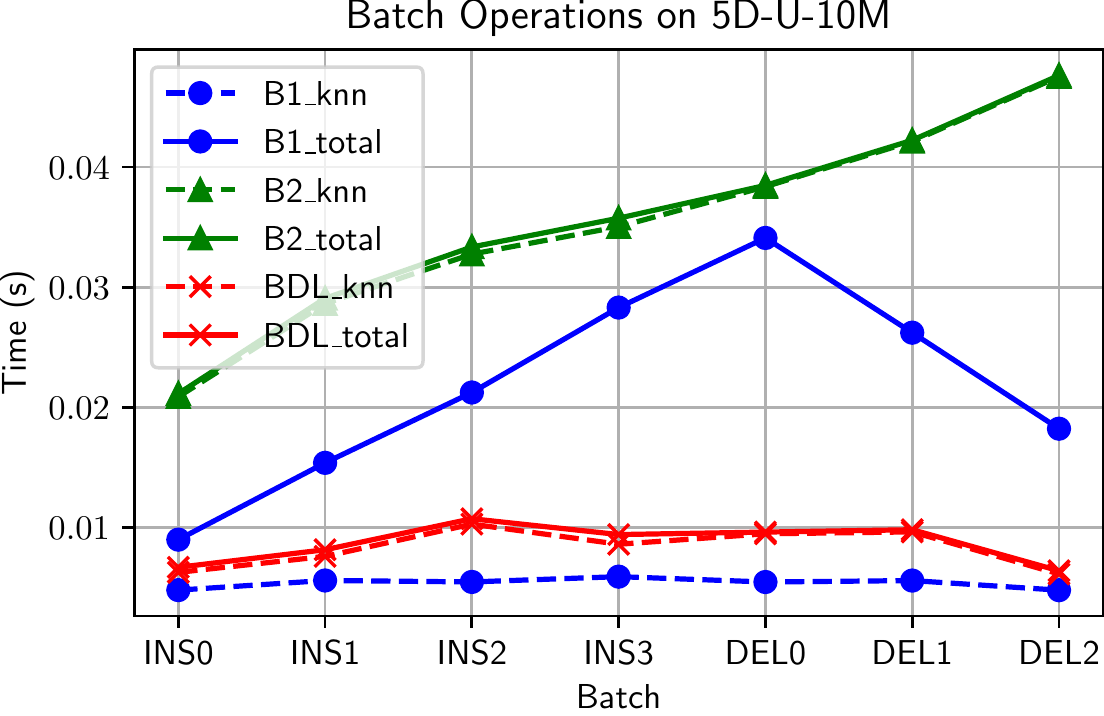}
        \caption{Batch updates and \knn{} queries on 7D-U-10M.}
        \label{fig:dynamic-ops-5du}
    \end{subfigure}
    \caption{Plots of running times (seconds) of updates and queries vs. progressive batch updates on
    the tree using all 36 cores with hyper-threading, for the 2D-V-10M and 7D-U-10M datasets. ``knn'' 
    represents the \knn{} query performance after cumulative updates on the trees and ``total'' 
    represents the combined running time of both updates and queries since the previous batch.}
\end{figure*}

\subsubsection{Dual-Tree \knn}\label{exp:dualtree-knn}
In addition to the data-parallel approach discussed above, we also implemented \knn{} based on the 
dual-tree traversal proposed by March~\etal~\cite{esmt}.
While we found our dual-tree \knn{} for \textbf{BDL} was faster on a single-thread due to
more efficient pruning, it was not as scalable as other \knn{} approaches and thus did not perform
as well in parallel. Therefore, we omit the experimental results for dual-tree \knn{}.

\section{Conclusion}
We have presented the \ourtree{}, a parallel batch-dynamic \kdtree{} that supports batched construction, insertions, deletions, and \knn{} queries. 
We show that our data structure has strong theoretical work and depth bounds. Furthermore, our experiments show that the \ourtree{} achieves good parallel speedup and presents a useful tradeoff between two baseline implementations. In particular, it delivers the best performance in a dynamic setting involving batched updates to the underlying dataset interspersed with \knn{} queries.

\section{Acknowledgements}
This research was supported by 
DOE Early Career Award \#DE-SC0018947,
NSF CAREER Award \#CCF-1845763, Google Faculty Research Award, Google Research Scholar Award, DARPA
SDH Award \#HR0011-18-3-0007, and Applications Driving Architectures
(ADA) Research Center, a JUMP Center co-sponsored by SRC and DARPA.

\bibliographystyle{IEEEtran}
\bibliography{refs}

\end{document}